\documentclass[10pt,a4paper]{article}

\usepackage[utf8]{inputenc}
\usepackage[T1]{fontenc}
\usepackage[english]{babel}
\usepackage{lmodern}

\usepackage[a4paper, left=20mm, right=20mm, top=2cm, bottom=2cm]{geometry}
\usepackage{setspace}
\addto\extrasenglish{%
}
\newcommand{\appendixnumberline}[1]{Appendix #1\quad}
\usepackage{amsmath,amsthm,mathdots,amssymb,bm,bbm,mathrsfs,mathtools}
\usepackage[final]{graphicx} % Remove the 'draft' option
\usepackage{subcaption}
\usepackage{tikz,pgfplots}\pgfplotsset{compat=1.16}
\usepackage{booktabs}

\usepackage{thm-restate}
\newtheorem{thm}{Theorem}%[section]
\newtheorem{cor}{Corollary}%[section]
\newtheorem{lem}{Lemma}%[section]

\newtheorem{prop}{Proposition}%[section]

\newtheorem{rmk}{Remark}%[thm]
\newtheorem{defn}{Definition}%[section]

\usepackage{comment}

\usepackage[colorlinks=true, urlcolor=violet, linkcolor=blue, citecolor=red, hyperindex=true, linktocpage=true, pagebackref=true, draft=false]{hyperref}
\usepackage{mleftright}\mleftright % fixes bracket spacing

\usepackage{enumitem}
\usepackage{microtype,color,graphicx}
\usepackage{tikz-cd}
\usepackage{xcolor}
\usepackage{multirow}
\usepackage{makecell}
\usepackage{tablefootnote}

\renewcommand{\vec}{\bm}
\newcommand{\CA}{\mathcal{A}}

\newcommand{\CE}{\mathcal{E}}
\newcommand{\BE}{\mathbb{E}}
\newcommand{\CF}{\mathcal{F}}

\newcommand{\CH}{\mathcal{H}}

\newcommand{\CL}{\mathcal{L}}

\newcommand{\CN}{\mathcal{N}}

\newcommand{\CO}{\mathcal{O}}

\newcommand{\BR}{\mathbb{R}}

\newcommand{\CT}{\mathcal{T}}

\newcommand{\CW}{\mathcal{W}}

\newcommand*{\veps}{\varepsilon}
\newcommand*{\veta}{\varepsilon}

\newcommand{\vA}{\bm{A}}

\newcommand{\vC}{\bm{C}}

\newcommand{\vH}{\bm{H}}
\newcommand{\vI}{\bm{I}}

\newcommand{\vK}{\bm{K}}

\newcommand{\vL}{\bm{L}}

\newcommand{\vO}{\bm{O}}
\newcommand{\vP}{\bm{P}}

\newcommand{\vU}{\bm{U}}

\newcommand{\vsigma}{\bm{ \sigma}}

\newcommand{\vrho}{\bm{ \rho}}

\newcommand{\tCO}{\tilde{\CO}}

\newcommand{\dagg}{\dagger}

\newcommand{\vertiii}[1]{{\left\vert\kern-0.25ex\left\vert\kern-0.25ex\left\vert #1 \right\vert\kern-0.25ex\right\vert\kern-0.25ex\right\vert}}
\newcommand{\norm}[1]{\Vert {#1} \Vert}

\newcommand{\labs}[1]{\left\vert {#1} \right\vert}

\newcommand{\ri}{\mathrm{i}}
\newcommand{\rd}{\mathrm{d}}
\newcommand*{\tr}{\mathrm{Tr}}
\newcommand*{\poly}{\mathrm{Poly}}

\newcommand{\nrm}[1]{\left\| #1 \right\|}

\newcommand{\undersetbrace}[2]{ \underset{#1}{\underbrace{#2}}}

\DeclarePairedDelimiterX{\braket}[1]{\langle}{\rangle}{#1}
\DeclarePairedDelimiterX\ketbra[2]{| }{|}{#1 \delimsize\rangle\!\delimsize\langle #2}	
\DeclarePairedDelimiterX\dotp[2]{\langle}{\rangle}{#1, #2}
\newcommand{\bigO}[1]{\mathcal{O}\left( #1 \right)}
\newcommand{\bigOt}[1]{\widetilde{\mathcal{O}}\left( #1 \right)}

\DeclareMathAlphabet{\dutchcal}{U}{dutchcal}{m}{n}
\SetMathAlphabet{\dutchcal}{bold}{U}{dutchcal}{b}{n}
\DeclareMathAlphabet{\dutchbcal} {U}{dutchcal}{b}{n}

\usepackage{ifdraft}
\ifdraft{
	\newcommand{\authnote}[3]{{\color{#3} {\bf  #1:} #2}}
}{
	\newcommand{\authnote}[3]{}
}

\usetikzlibrary{quantikz2}

\begin{document}

\title{Efficiently estimating quantum thermal properties\\ from exponentially fewer samples}
	\author{
		Chi-Fang Chen\thanks{EECS, University of California, Berkeley, CA 94720, USA \&\\
		Center for Theoretical Physics, Massachusetts Institute of Technology, Cambridge, MA 02139, USA.
		{\tt achifchen@gmail.com}}
		\and 
		András Gilyén\thanks{Alfréd Rényi Institute of Mathematics, Budapest, Hungary. 
		{\tt gilyen@renyi.hu}}
	}
\maketitle

\begin{abstract}
   We present an efficient protocol for estimating $M$ observables to precision $\varepsilon$ from only $\mathcal{O}(\log (M)/\varepsilon^2)$ copies of a Gibbs state $\bm{\rho}\propto\exp(-\beta \bm{H})$, given access to its Hamiltonian $\bm{H}$. The protocol uses single-copy, nonadaptive measurements and a total Hamiltonian simulation time of $\widetilde{\mathcal{O}}(\beta M/\varepsilon^2)$; we show that the sample complexity is optimal in a black-box setting. The key idea of our protocol is a new interpretation of quantum Gibbs samplers as \textit{detailed-balanced measurement channels}: measurements that preserve the Gibbs state when outcomes are marginalized, and thus enable Gibbs samples to be reused for measuring subsequent observables. Consequently, shadow tomography of thermal states admits a general efficient algorithm when controlled Hamiltonian simulation is available, substantially boosting the data acquisition rate in quantum thermal simulation workflows.
\end{abstract}

\setcounter{tocdepth}{2}
\tableofcontents
\newpage

\section{Introduction}

Simulation of quantum many-body systems is widely expected to be a major application of future quantum computers. Consequently, preparing quantum thermal and ground states—central tasks in this vision—has motivated a broad array of sophisticated quantum algorithm proposals~\cite{temme2009QuantumMetropolis,shtanko2021PreparingQThermalStatesNoieslessly,chen2025EfficientQThermSim,ding2024GibbsSamplingViaKMS,jiang2024QMetropolisWeakMeas,ding2024SingleAncillaGSPrepViaLindblad,rall2023ThermalStatePrepRounding} as well as experimental demonstrations~\cite{mi2024stable}.

Nevertheless, many end-to-end applications in chemistry and materials science, especially when integrated with machine learning workflows, ultimately require estimating a large collection of observables. Because quantum measurements are inherently destructive, this task leads to a rich set of fundamental problems in quantum state tomography~\cite{haah2017OptTomography,odonnell2016EfficientQuantumTomography,knill2007optimal,huggins2021QAlgMultipleExpectationValues, bonet2020nearly,zhao2021fermionic,wan2023matchgate,king2025TriplyEffShadowTomo,pelecanos2025MixedStateTomographyReduces}. 
Despite important progress on some aspects of this problem~\cite{aaronson2018ShadowTomography,huang2020PredictingManyProps}, the best general-purpose efficient protocol~\cite{sinha2025DinIndepEfficientShadowTomo}
for estimating many observables only achieves a quadratic improvement over the straightforward measure, discard, and repeat procedure that uses an extensive number of samples.
This heavy \textit{multiplicative} overhead from tomography threatens to narrow the regime of quantum advantage in quantum simulation, especially because the competing classical approaches (e.g., tensor networks \cite{schollwock2011density,orus2019tensor,verstraete2023density}) typically provide an explicit description of the samples, which often enables extracting observables with an \textit{additive} cost. 

In this work, we introduce a general measurement protocol for estimating many thermal observables which is:
\begin{itemize}
\item \textbf{Sample efficient.} The number of thermal samples scales logarithmically with the number of observables and quadratically with the desired precision.
\item \textbf{Computationally efficient.} The protocol uses single-copy measurements. The total (controlled) Hamiltonian evolution time scales linearly with the inverse temperature and the number of observables, and quadratically with the precision.
\end{itemize}

Generally, shadow tomography~\cite{aaronson2018ShadowTomography} has shown that many observables can be estimated information-theoretically from a small number of copies of an unknown quantum state. However, it relies on entangled operations across many copies and appears computationally inefficient. Classical shadow methods with single-copy measurements can provide computationally efficient protocols, but only for restricted classes of observables, such as low-rank or few-body operators~\cite{huang2020PredictingManyProps}. Further developments in this direction have shown~\cite{king2025TriplyEffShadowTomo}, that using two-copy measurements one can efficiently estimate a broader class of observables, however the corresponding protocols are more involved and often feature degraded precision dependence. By contrast, our single-copy measurement protocol matches the corresponding tight classical sample-complexities, and can be efficiently implemented for arbitrary bounded and efficiently measurable observables, including Green’s functions and high-weight operators that commonly arise in quantum simulation. 

The key idea behind our protocol is a new use of recently constructed quantum Gibbs samplers as 
\begin{align*}
\textit{measurement channels that preserve the Gibbs state when outcomes are ignored.}
\end{align*}
Operationally, the measurement protocol can be viewed as running a carefully designed Gibbs sampling algorithm, where the Kraus operators are tailored to the desired observables. Conceptually, the mechanism for estimating multiple observables resembles making statistical inferences from trajectories of classical Markov chains. Our results show that the destructive nature of quantum measurements can be effectively circumvented by exploiting quantum detailed balance, and the quantum trajectories can give surprisingly powerful statistics. Even though our work focuses entirely on tomography and keeps the thermal state preparation as a black box, we expect the interplay between the algorithmic and tomographic nature of quantum Gibbs samplers to be a fertile ground~\cite{jiang2026PredictingPropQThermalStates}.

More broadly, our protocol suggests that many black-box tomography problems admit fundamentally different solutions when additional structure of the initial state—such as knowledge of the Hamiltonian—is available~\cite{chen2025CatalyticTomographyGroundStates}. In such cases, this structure can translate into orders-of-magnitude reductions in state-preparation costs in practical quantum simulation workflows, and substantially accelerate acquisition of quantum data at scale.

\subsection{Prior work}

Due to the destructive nature of quantum measurements, the tomography of general unknown quantum states has become a rich and rapidly evolving field. A major breakthrough was the proposal of shadow tomography~\cite{aaronson2018ShadowTomography}, which showed that $M$ (spectrally bounded) observables can be estimated using only polylogarithmically many copies of an unknown quantum state.\footnote{\label{foot:ImprovedShadowSampleCompl}A few months after our manuscript was posted, two preprints further improved the sample complexity of general shadow tomography. The first improves the sample complexity to $\CO\big(\frac{\log(M)\sqrt{\log(d)}}{\varepsilon^3}\big)$~\cite{chen2026OnlineShadowTomographyOpt} (for a small fixed $\delta$) and works even in an adversarial online setting, while the second claims a sample complexity of $\bigO{\frac{\log(M)\log(M/\delta)}{\varepsilon^2}}$~\cite{jeronimo2026DimensionFreePolylogQShadowTomo} -- removing any dimension dependence. It is unknown whether these protocols can be implemented in a computationally efficient way.} However, the original formulation focuses primarily on information-theoretic learnability. Implementing the protocol appears computationally challenging, as it relies on entangled measurements across many copies and classical post-processing steps such as matrix multiplicative weight updates of exponentially large matrices. It remains unclear whether these procedures admit efficient algorithms.

Subsequently, the widely celebrated framework of classical shadows~\cite{huang2020PredictingManyProps} demonstrated that the central promise of shadow tomography can be achieved efficiently for certain classes of observables, particularly low-rank or low-weight operators. An appealing feature of this protocol is that the observables need only be specified after the measurements are performed. However, the sample complexity depends on a quantity known as the shadow norm, which can grow prohibitively large for high-rank or high-weight observables. 
More recently, a rapidly growing family of shadow-based tomography protocols established non-trivial tradeoffs between sample complexity and the class of observables that can be estimated~\cite{chen2024OptimalEntanglementSampleTradeoffStateTomo,king2025TriplyEffShadowTomo}. Nevertheless, it remains open whether a sample complexity logarithmic in the number of observables can be achieved for arbitrary observables.

A complementary line of work seeks polynomial improvements over the naive $M/\varepsilon^2$ scaling in a setting where a purification of the quantum state can be prepared in a coherent fashion, enabling the use of advanced quantum algorithmic techniques~\cite{huggins2021QAlgMultipleExpectationValues,apeldoorn2022QTomographyWStatePrepUnis}. While these approaches can provide non-trivial quadratic speedups, it seems unlikely that larger improvements are achievable using similar methodology.

Our protocol offers a different conceptual route for reducing tomography costs: implicit knowledge about the initial state. In many quantum simulation settings, the Hamiltonian is known, even though preparing the corresponding Gibbs state may be computationally hard. We show that this structure can be exploited to effectively reuse Gibbs samples by designing measurement protocols that satisfy (KMS) quantum detailed balance. Concretely, we construct a Gibbs sampler whose Kraus operators are associated with the desired observables, allowing measurements to be performed while preserving the Gibbs state when outcomes are ignored. The key distinction from the general shadow tomography problem is that the thermal setting is fundamentally computational: information-theoretically, the Gibbs state is already specified by the Hamiltonian. Remarkably, access to the Hamiltonian alone suffices to achieve shadow-tomography-like sample complexity while avoiding the destructive overhead coming from extensive repeated state preparation. 

\begin{table}[hb!]
\centering
\renewcommand{\arraystretch}{2.5}
\begin{tabular}{l c c c}
\hline
Method & Sample Complexity & Efficient & Observable Class \\
\hline
Classical random variables
& $\displaystyle \Theta\!\bigg(\frac{\log(M/\delta)}{\varepsilon^2}\bigg)_{\phantom{f}}$ 
& yes 
& $|O_i|\le 1$ \\
\hline\hline
Thermal samples -- \autoref{thm:main} \& \ref{thm:tightSample}
& $\displaystyle \Theta\!\bigg(\frac{\log(M/\delta)}{\varepsilon^2}\bigg)$ 
& yes, $\displaystyle \bigOt{\frac{\beta M}{\varepsilon^2}}$ 
& $\|\vA_i\|\le 1$ \\

Naive
& $\displaystyle \bigO{\frac{M\log(M/\delta)}{\varepsilon^2}}$ 
& yes 
& $\|\vA_i\|\le 1$ \\

Best general \& efficient protocol~\cite{sinha2025DinIndepEfficientShadowTomo}
& $\displaystyle \CO\bigg(\frac{\sqrt{M}\log(M/\delta)}{\varepsilon^2}\bigg)$
& yes
& $\|\vA_i\|\le 1$ \\

Shadow tomography~\cite{aaronson2018ShadowTomography,badescu2024ImprovedQDataAnalysis} 
& $\displaystyle \bigOt{\frac{\log^{2}(M)\,\log(D)}{\varepsilon^{4}}}$ 
& unknown 
& $\|\vA_i\|\le 1$ \\

Classical shadows 
& $\displaystyle \bigO{\frac{\log(M/\delta)}{\varepsilon^2}\,c_{\text{shadow}}}_{\phantom{f}}$
& sometimes 
& \makecell{low-rank or \\ low-weight} \\
\hline
\end{tabular}
\caption{
Comparison of measurement strategies for estimating $M$ observables. The sample complexity of our \autoref{thm:main} is tight under the assumption that we can only use controlled Hamiltonian simulation for a subexponential amount of time, see \autoref{thm:tightSample}.
The naive algorithm uses $\bigO{\frac{\log(M/\delta)}{\varepsilon^2}}$ fresh Gibbs samples for each of the $M$ observables. The improved sample complexity of shadow tomography\textsuperscript{\ref{foot:ImprovedShadowSampleCompl}} is proven in~\cite{badescu2024ImprovedQDataAnalysis} and that of classical shadows can be found in~\cite{huang2020PredictingManyProps}. When two-copy measurements are also allowed, high-weight Pauli operators can be also efficiently estimated~\cite{king2025TriplyEffShadowTomo}, however with a quartic precision dependence.}
\end{table}

\subsection{Problem setup}\label{sec:setup}
Given an $n$-qubit Hamiltonian $\vH$ and inverse temperature $\beta \ge 0$, consider the Gibbs state
\begin{align*}
\vrho_{\beta} = \frac{e^{-\beta \vH}}{\tr[e^{-\beta \vH}]} .
\end{align*}
We are interested in bounded Hermitian observables
\begin{align*}
\vA_1,\ldots,\vA_M \quad \text{such that}\quad \vA_i= \vA_i^{\dagger},\quad \norm{\vA_i} \le 1 \quad \text{for each}\quad 1\le i \le M.
\end{align*}
Our goal is to estimate all thermal expectations
\begin{align*}
\tr[\vrho_{\beta} \vA_i] \in [-1,1]\quad \text{to precision}\quad \varepsilon,
\end{align*}
{with failure probability} at most $\delta$ on a (fault-tolerant) quantum computer. Even though our presentation will assume full digital control, we believe the key methodology shall be transferable to analog settings.

We make the following assumptions:
\begin{itemize}

\item \textbf{(Gibbs samples.)} We assume (black-box) access to \textit{independent} copies of the Gibbs state $\vrho_{\beta}$.

\item \textbf{(Hamiltonian simulation.)} Our protocol requires (black-box) access to controlled Hamiltonian time evolutions $e^{\pm i\vH t}$. 

\item \textbf{(Access to observables.)} We assume block-encoding access to each observable $\vA_i$. This assumption is automatically satisfied when $\vA_i$ is implemented by a self-adjoint unitary circuit, such as a Pauli string.

\end{itemize}
To keep our results clean, we use the asymptotic notation $\widetilde{\CO}(T)=\CO\!\left(T\poly\!\log(T + M + \beta + \nrm{\vH} + \frac{1}{\veps} + \frac{1}{\delta})\right)$.

\subsection{Main results}

Our main result is a general measurement algorithm with provable guarantees.
\begin{thm}[Measuring many observables with very few samples of Gibbs states]
\label{thm:main}
    Consider $M$ Hermitian observables $\vA_1, \cdots, \vA_M$ such that $\norm{\vA_i}\le 1$, a Hamiltonian $\vH$, and inverse temperature $\beta \ge 0.$ Then, we can estimate all thermal expectations $\tr[ \vrho_\beta\vA_{i}]$ within error $\varepsilon$, with failure probability at most $\delta$ using
    \begin{align*}
    S &= \CO\left(\frac{\log(M/\delta)}{\varepsilon^2}\right)\quad \text{samples of the Gibbs state}\quad \vrho_\beta,\\
    b &=  \bigOt{\frac{M}{\varepsilon^2}} \quad \text{total uses of block encodings of}\quad \{\vA_i\}\text{, and}\\
    t &= \bigOt{\frac{\beta M}{\varepsilon^2}} \quad \text{total Hamiltonian simulation time}.
    \end{align*}
\end{thm}
In fact, the observables can be given in a (non-adversarial) \textit{online} fashion, the measurement scheme is nonadaptive, and only uses a single copy of the quantum Gibbs state at a time. We only displayed the Hamiltonian simulation time and the total number of observable block-encoding uses, as these together dominate the cost of the quantum computation. Here we assume that one is given access to the bounded observable $\vA_i$ via a unitary block-encoding, since any efficient implementation of a single-copy estimation procedure can be turned into an efficient block encoding~\cite[Definition~2 and Equation~14]{apeldoorn2018ImprovedQSDPSolving}.

Alternatively, if all observables are known ahead of time, we can also perform all measurements on a single $n$-qubit register, as long as we are willing to wait for a mixing time to get another fresh Gibbs sample. In fact, one can use the Gibbs state as a ``warm-start'' for preparing the next Gibbs sample~\cite{jiang2026PredictingPropQThermalStates}, where the ``autocorrelation'' time can be faster than a mixing time from worst-case initial states, as also noted by~\cite{chen2026ThermalExpectationsSingleTraj}.

\paragraph{Proof ideas.} Recently, a new wave of quantum algorithms emerged for preparing thermal states using detailed-balanced Lindbladians or quantum channels~\cite{temme2009QuantumMetropolis,shtanko2021PreparingQThermalStatesNoieslessly,rall2023ThermalStatePrepRounding,chen2025EfficientQThermSim, ding2024SingleAncillaGSPrepViaLindblad,ding2024GibbsSamplingViaKMS,gilyen2024QGeneralGlauberMetropolis}. Operationally, these algorithms can be viewed as implementing carefully controlled system–ancilla interactions. When the ancillae are traced out, the resulting effective dynamics can be designed to satisfy an exact structural property—quantum detailed balance~\cite{chen2023ExactQGibbsSampler}—that guarantees the Gibbs state to be a fixed point.

The key idea of our protocol is a conceptual extension of these quantum Gibbs-sampling algorithms. Rather than \textit{discarding} the ancilla, we realize that the ancilla—recording which transition or “jump’’ has occurred—already functions as a metrological device. This perspective leads us to introduce the central primitive of our algorithm: detailed-balanced measurement channels (\autoref{fig:measurement-gadget}).

\begin{figure}[t]
\centering
\begin{tikzpicture}[x=1.2cm,y=1.0cm,>=latex]

% --- CN (top) ---
\node at (0,0) {$\vrho_\beta$};
\draw[->] (0.3,0) -- (1.0,0);

\node[draw,minimum width=1.0cm,minimum height=0.7cm] (CN) at (1.8,0) {$\CN$};

\draw[->] (2.3,0) -- (3.2,0);
\node[right] at (3.2,0) {$\vrho_\beta$};

% dashed outcome leg
\draw[dashed,thick,->] (1.8,-0.35) -- (2.2,-0.75);
\node[right] at (2.2,-0.75) { };

% --- CE (bottom) ---
\node at (0,-1.6) {$\vrho_\beta$};
\draw[->] (0.3,-1.6) -- (1.0,-1.6);

\node[draw,minimum width=1.0cm,minimum height=0.7cm] (CE) at (1.8,-1.6) {$\CE$};

\draw[->] (2.3,-1.6) -- (3.2,-1.6);
\node[right] at (3.2,-1.6) {$\vK_i \vrho_\beta \vK_i^\dagger$};

% solid outcome leg
\draw[thick,->] (1.8,-1.95) -- (2.2,-2.35);
\node[right] at (2.2,-2.35) {$i$};

\end{tikzpicture}

\caption{The averaged channel $\CN$ (top)  preserves $\vrho_\beta$ when the measurement outcomes are marginalized. The full measurement channel $\CE$ (bottom) outputs the post-measurement state conditioned on outcome $i$.}
\label{fig:measurement-gadget}
\end{figure}
\begin{defn}[Detailed-balanced measurement channel]
Given a Hermitian observable $\vA$ and a full-rank quantum state $\vrho$, we say a collection of Kraus operators $\{\vK_i\}$ with $\sum_i \vK_i^{\dagger}\vK_i=\vI$ forms a $\vrho$-detailed-balanced $w$-weak measurement channel for $\vA$ if:
\begin{itemize}
\item (detailed-balanced on average) When the measurement outcomes are marginalized, the resulting channel $\CN[\cdot] = \sum_i \vK_i[\cdot]\vK_i^{\dagger}$ satisfies (KMS) detailed balance with respect to $\vrho$, implying $\CN[\vrho]=\vrho$.
\item (informative outcomes) A bounded linear combination of the outcome probabilities $\braket{\vK_i^{\dagger}\vK_i}_{\vrho}$ provides an unbiased estimator of $\braket{\vA}_{\vrho}$, i.e., $\braket{w\vA}_{\vrho}=\sum_i \lambda_i \braket{\vK_i^{\dagger}\vK_i}_{\vrho}$ for some known coefficients $\lambda_i\in[-1,1]$, and $$\Bigg\lVert\sum_{i\colon \lambda_i\neq 0}\vK_i^{\dagger}\vK_i\Bigg\rVert=\bigO{w}.$$
\end{itemize}
\end{defn}

Ideally, we would like the ``weakness'' parameter $w$ to be as large as possible. Our current construction yields Kraus operators that are observed with probability scaling logarithmically with the system size and the inverse temperature. To retain the optimal sample complexity we carefully design a readout scheme by repeating these measurements $\bigO{1/w}$ many times per sample. In the \hyperref[sec:Methods]{Methods} we also describe a continuous-time Lindbladian counterpart of the above discrete measurement channel that in some sense circumvents this weakness.

It will be useful to distinguish between the averaged channel and the full measurement apparatus (including the appended ancillae):
\begin{align*}
\CE[\vsigma] := \sum_i \vK_i \vsigma \vK_i^{\dagger} \otimes \ketbra{i}{i},
\qquad
\text{such that}\qquad
\tr_{anc}[\CE] = \CN .
\end{align*}

Using the above $\vrho$-detailed-balanced measurement channel as a black-box primitive, we can sketch the main estimation algorithm (\autoref{fig:parallel-measurements}). Consider performing the following sequential measurement procedure $r$ times on $r$ independent copies of $\vrho$:
\begin{align*}
\CE^{(1)}_{M}\cdots\CE^{(1)}_{1}[&\vrho]
\quad \text{such that}\quad
\CN^{(1)}_{M}[\vrho]=\cdots=\CN^{(1)}_{1}[\vrho]=\vrho ,\\
&\vdots\\
\CE^{(r)}_{M}\cdots\CE^{(r)}_{1}[&\vrho]
\quad \text{such that}\quad
\CN^{(r)}_{M}[\vrho]=\cdots=\CN^{(r)}_{1}[\vrho]=\vrho .
\end{align*}

For each copy, we record the outcome of every $\CE_m^{(j)}$.
To estimate the observable $\vA_m$, we only consider the statistics of the outcome labels in the $m$-th column, i.e., marginalize over all other outcomes. Due to the detailed-balance property, the distribution of the outcomes in the $m$-th column (tracing out the previous measurement history) is identical to the distribution obtained by directly measuring on the Gibbs state:
\begin{align*}
\CE^{(1)}_{m}\CN^{(1)}_{m-1}\cdots \CN^{(1)}_{1}[\vrho]
&= \CE^{(1)}_{m}[\vrho],\\
&\,\,\vdots\\
\CE^{(r)}_{m}\CN^{(r)}_{m-1}\cdots \CN^{(r)}_{1}[\vrho]
&= \CE^{(r)}_{m}[\vrho].
\end{align*}

The quadratic $1/\varepsilon^2$ scaling arises from standard mean estimation, while the $\log(M/\delta)$ dependence follows from a union bound over the tail bounds of the $M$ (possibly correlated) mean estimates. For technical reasons our efficiently implementable detailed-balanced measurements constructed from Gibbs samplers actually provide a weak measurement $\CE_i(w)$, that only produces a meaningful outcome with probability $\bigO{w}$, where $w$ is known and scales as $\sim \log^{\!-2}(\beta\nrm{\vH})$. Thus, a single application of $\CE_i(w)$ can only provide a bounded unbiased estimator of $\braket{w\vA}_{\vrho}$. We denote by $\CE_i$ the repeated application of $\CE_i(w)$ about $2/w$ times, which boosts this to a constant variance unbiased estimator of $\braket{\vA}_{\vrho}$, while using the \emph{same copy} of the Gibbs state (see \autoref{lem:measurementTailBound} in \autoref{apx:repeatTailBond}). Using $\CE_i$ thus enables retaining a $w$-independent optimal sample complexity.

Intriguingly, the sequential measurement outcomes for different observables on the \textit{same} Gibbs sample can be highly correlated.
Indeed, conditioned on observing any previous outcome, the post-selected quantum state may become very far from Gibbs.\footnote{Actually, this is the case for the family of Hamiltonians and observables that we use in our sample complexity lower bound, the proof of \autoref{thm:tightSample}.} Nevertheless, for the purposes of estimating individual observables, the post-measurement state can be treated effectively as the Gibbs state ``on average'', when we only consider the empirical measurement mean from sufficiently many $(\sim \log(M))$ independent copies.

\paragraph{Tight sample bounds.} We show that our main result \autoref{thm:main} has an optimal sample complexity (up to a constant factor) in a setting where one is given black-box access to the system Hamiltonian $\vH$, and we are only allowed to perform black-box Hamiltonian simulation for a sub-exponential amount of time (in terms of the number of qubits); this is a standard lower bound setting in distribution learning. Our lower bound, formalized in \autoref{thm:tightSample}, holds even for classical Hamiltonians. 

\begin{restatable}{thm}{sampleLowerBound}\label{thm:tightSample}
	For all $\veps,\delta\in (0,\frac{1}{10})$ and $M\geq 1$ there is a set of $n$-qubit observables $\{\vO_i\colon i\in [M]\}$ of operator norm at most $1$ for some $n=\Omega\big(\log(\frac{\log(M+1)}{\veps\delta})\big)$, a family of classical Hamiltonians $\vH\in\{0,1\}^{2^n\times 2^n}$, and an inverse temperature $\beta=\bigO{n}$,
	for which estimating to precision within $\veps$ the expectation of all observables $\vO_i$ with success probability at least $1-\delta$ requires $\Omega\big(\frac{\log(M/\delta)}{\veps^2}\big)$ Gibbs samples in the worst case. This lower bound holds even if we are allowed to make $2^{\frac{n}{3}}$ (controlled) quantum queries to the phase oracle $\ket{j}\rightarrow(-1)^{\vH_{\! jj}}\ket{j}$.
\end{restatable}

Note that for simplicity we stated the above lower bound in terms of discrete (phase) queries. However, this also implies qualitatively the same bound when we consider the total Hamiltonian simulation time as the complexity measure, because a fractional query algorithm can be efficiently simulated by discrete queries~\cite{cleve2008EfficientDiscContQuery}, e.g., to trace distance $\delta$ by using $\bigO{T\log(T/\delta)}$ discrete queries~\cite[Theorem 1.3]{berry2013ExpPrecHamSimSTOC}. 
Indeed, assume towards a contradiction that by allowing a total simulation time of $T=2^{\frac{n}{4}}$, it suffices to take $o\big(\frac{\log(M/\delta)}{\veps^2}\big)$ Gibbs samples. Then this protocol could be simulated to error at most $\delta$ with $T\log(T/\delta) = 2^{\frac{n}{4}} \bigO{n^2} \ll 2^{\frac{n}{3}}$ discrete phase queries (and still $o\big(\frac{\log(M/\delta)}{\veps^2}\big)$ Gibbs samples), which contradicts \autoref{thm:tightSample}.

Our proof's core idea is a reduction\footnote{While our lower bound \autoref{thm:tightSample} is stated as a worst case result, via our reduction it implicitly inherits average case hardness from the distribution learning lower bound of \cite[Lecture 4]{diakonikolas2019AdvancedLearningTheoryLectureNotes} corresponding to the following explicit distribution of Hamiltonians. 
Let $B_k\subset \{0,1\}^n$ denote the $k$-th bin of equal size $2^{n-\mu}$ for $\mu=\lfloor\log_2\lfloor\log_2(M+1)+1\rfloor\rfloor$, and let $b\in\{-1,1\}^{\frac m2}$ be uniformly random where $m=2^{\mu}$. Let $\vH_{jj}:=f(j)$, $w:=2\left\lfloor\frac{\delta^2}{8m^2}2^{\frac n3}\right\rfloor$, and $f\colon \{0,1\}^n\rightarrow\{0,1\}$ a uniformly random Boolean function of Hamming weight $2^n-\frac{mw}{2}$ such that $\forall i\in[\frac{m}{2}]$ we have $|f^{-1}(0)\cap (B_{2i}\cup B_{2i-1})|=w$ and $|f^{-1}(0)\cap B_{2i}|-|f^{-1}(0)\cap B_{2i-1}|=b_i\lceil \veps w\rceil$. The $M\approx2^{m}$ observables are simple indicators over unions of the $m$ bins $B_k$ as in the proof of \autoref{thm:tightSample}.} to the known sample complexity lower bound $\Omega\left(\frac{m+\log(1/\delta)}{\veps^2}\right)$ for learning a classical distribution on $m$ elements to total variation distance $\veps$ with success probability at least $1-\delta$ \cite[Theorem 1.3]{cannone2022DistributionTestingSurvey},\cite[Lecture 4]{diakonikolas2019AdvancedLearningTheoryLectureNotes}. We apply this lower bound to a distribution on $m\approx\log(M)$ elements induced by $m$ bins of the configuration space featuring a specific zero-temperature Gibbs state corresponding to a classical $\{0,1\}$-valued Hamiltonian. Observing that $\log(M/\delta)=\Theta(m + \log(1/\delta))$, this implies the above theorem by simply setting the $M$ observables to the $2^m$ indicators corresponding to all possible combinations of the $m$ bins. Indeed, it is easy to see that learning all these $2^m$ indicators to $\varepsilon$ precision also implies learning the distribution induced on the $m$ bins to total variation distance $2\varepsilon$.

Our specific zero-temperature Gibbs distribution is fairly close to a Gibbs state at a reasonably low temperature (see \autoref{lem:tempTVBound} in \autoref{apx:sampLower}), for which therefore the same lower bound holds. Nevertheless, it is conceptually cleaner for us to reduce to the zero-temperature case, because that simply corresponds to a uniform distribution over zero-energy states

We choose a classical (diagonal) Hamiltonian with less than $2^{\frac{n}{3}}$ zero-energy configurations, and since the spectrum is binary, after seeing a sample, querying that element does not reveal any new information. Also because the zero-energy states are exponentially rare (have frequency about $\sim 2^{\frac{-2n}{3}}$) it requires an exponential amount of queries to obtain any previously unseen sample. This tells us that the query access reveals very little information, and therefore the $\Omega\left(\frac{m+\log(1/\delta)}{\veps^2}\right)$ sample lower bound holds (see \autoref{cor:stringSampLower} in \autoref{apx:sampLower}), unless an exponential number of queries are used. We formally quantify the above quantum query complexity argument using a low-probability version of the hybrid lower bound method (see \autoref{lem:smallHybLow} in \autoref{apx:sampLower}).

\section{Discussion}

To overcome the apparent measurement overhead in quantum simulation, we designed a protocol that measures $M$ observables using only $\log(M)$ thermal states. The protocol only requires the observables to be spectrally bounded (which captures most physical observables, such as Green's functions), can be efficiently implemented on a quantum computer with access to (controlled) time-evolution by the system Hamiltonian, and is sample-optimal in a certain black-box model. We believe our protocol could play an important role in practical applications by enhancing the rate of data acquisition in thermal simulation workflows. In turn, this could fundamentally accelerate quantum data-intensive chemistry and materials applications in broader settings.

Broadly, our results provide a novel way of reusing complex quantum states by designing measurements that statistically preserve the given mixed quantum state. We briefly compare our results to other related developments as follows. Recently, quantum Gibbs states (and approximate stationary states of detailed-balanced Lindbladians) have been proved to satisfy a local Markov property: perturbations coming from local noise can be locally removed by the natural Lindbladian dynamics~\cite{chen2025QGibbsStatesLocalMarkov,bergamaschi2025StructuralTheoryQMetastability}. However, this Markov property turns out to be unnecessary to our shadow tomography statement; we do not need to run any additional Lindbladian between our detailed balanced measurements, provided access to multiple independent Gibbs states. To our knowledge, so far the Markov property has only been shown to be useful for measurements with an additional assumption of correlation decay, where a single copy of the state can produce effectively independent measurement results by interspersing measurements with Lindbladian dynamics~\cite{bergamaschi2025StructuralTheoryQMetastability,chen2026NoteStrongQMarkovProps}. 

In the case of estimating properties of a gapped ground state --- another natural target of quantum simulation --- a single copy of the state is sufficient for measuring observables~\cite{farhi2010quantum,chen2025CatalyticTomographyGroundStates} regardless the precision and the number of observables. However, this is unique to the pure state setting, and the efficiency depends on the corresponding spectral gap. By contrast, our shadow tomography protocol for the Gibbs state only requires the efficient simulation of the system Hamiltonian on top of providing independent quantum Gibbs samples.

Extending efficient shadow tomography beyond thermal states brings additional caveats. Indeed, our protocol heavily relies on the exact interplay between the Gibbs state, detailed balance, and knowledge of the Hamiltonian.  
For (approximate) stationary states~\cite{guo2025designing} of generic Lindbladians that do not correspond to sensible thermal states, however, much less is currently known both from an algorithmic and a theoretical perspective.
Indeed, when the measurement channel no longer preserves the target state, the $\log(M)$ scaling may be lost, as each measurement can accumulate systematic deviation from the target state.

That said, we believe our protocol shall be extendable to approximate stationary states of Gibbs samplers~\cite{bergamaschi2025StructuralTheoryQMetastability}, although one likely needs to restrict to observables that preserve the approximate stationary state to a sufficiently high extent.  
In the case where the Hamiltonian is only known approximately, or the algorithmic implementation is noisy --- such as when interfacing with analog samples --- the number of observables one can measure on the same state will depend on the errors. Even though our presentation focuses on fault-tolerant, digital simulation settings, we believe the idea of detailed-balanced measurements will be a valuable component in broader settings and may be instantiated in many different ways in practical quantum simulation workflows. 

\section{Methods}\label{sec:Methods}

The main algorithmic result is instantiating the detailed-balanced measurement channel via a quantum Gibbs sampler. Conceptually, it may be cleaner to think in discrete time, but we also introduce the continuous-time version for completeness. Once such primitives are instantiated using an off-the-shelf Gibbs sampling framework (with some specification), the proof of the main results (\autoref{thm:main}) follows naturally.

\subsection{Efficiently implementable detailed-balanced measurement channels}

A natural instantiation of a detailed-balanced measurement channel is a detailed balanced Gibbs sampler with carefully designed Kraus operators. Our construction builds on the discrete-time Gibbs sampler of~\cite{gilyen2024QGeneralGlauberMetropolis} which constructs a single complementary Kraus operator, that is a detailed-balanced quantum analogue of the ``reject'' step in the classical Metropolis algorithm. Unlike in the classical case, this additional Kraus operator depends non-trivially and non-linearly on the ``transition'' part; to bound the non-linearity and to efficiently construct this Kraus operator~\cite{gilyen2024QGeneralGlauberMetropolis} requires a mild ``weakness'' condition on the ``transition'' part which is the reason why we need to resort to using somewhat weaker than ideal measurements.

\begin{figure}[t]
\centering
\begin{tikzpicture}[x=1.2cm,y=0.8cm]
\newcommand{\cout}{0.35}

% Row 1
\node at (0,0) {$\vrho$};
\draw[->] (0.3,0) -- (1,0);

\node[draw,minimum width=1.0cm,minimum height=0.6cm] (E11) at (1.8,0) {$\CE^{(1)}_1$};
\node at (3.0,0) {$\cdots$};
\node[draw,minimum width=1.0cm,minimum height=0.6cm] (E1i) at (4.2,0) {$\CE^{(1)}_i$};
\node at (5.4,0) {$\cdots$};
\node[draw,minimum width=1.0cm,minimum height=0.6cm] (E1M) at (6.6,0) {$\CE^{(1)}_M$};

\draw[->] (2.3,0) -- (2.7,0);
\draw[->] (3.3,0) -- (3.7,0);
\draw[->] (4.7,0) -- (5.1,0);
\draw[->] ($(E1M.west)+(-0.4,0)$) -- (E1M.west);

\draw[->] (E11.south) -- +(\cout,-\cout) node[right] {$s^{(1)}_1$};
\draw[->] (E1i.south) -- +(\cout,-\cout) node[right] {$s^{(1)}_i$};
\draw[->] (E1M.south) -- +(\cout,-\cout) node[right] {$s^{(1)}_M$};

% dots row
\node at (0,-1.8) {$\vdots$};
\node at (1.8,-1.8) {$\vdots$};
\node at (4.2,-1.8) {$\vdots$};
\node at (6.6,-1.8) {$\vdots$};

% Last row
\node at (0,-3.6) {$\vrho$};
\draw[->] (0.3,-3.6) -- (1,-3.6);

\node[draw,minimum width=1.0cm,minimum height=0.6cm] (Er1) at (1.8,-3.6) {$\CE^{(r)}_1$};
\node at (3.0,-3.6) {$\cdots$};
\node[draw,minimum width=1.0cm,minimum height=0.6cm] (Eri) at (4.2,-3.6) {$\CE^{(r)}_i$};
\node at (5.4,-3.6) {$\cdots$};
\node[draw,minimum width=1.0cm,minimum height=0.6cm] (ErM) at (6.6,-3.6) {$\CE^{(r)}_M$};

\draw[->] (2.3,-3.6) -- (2.7,-3.6);
\draw[->] (3.3,-3.6) -- (3.7,-3.6);
\draw[->] (4.7,-3.6) -- (5.1,-3.6);
\draw[->] ($(ErM.west)+(-0.4,0)$) -- (ErM.west);

% outputs
\draw[->] (Er1.south) -- +(\cout,-\cout) node[right] (s1) {$s^{(r)}_1$};
\draw[->] (Eri.south) -- +(\cout,-\cout) node[right] (si) {$s^{(r)}_i$};
\draw[->] (ErM.south) -- +(\cout,-\cout) node[right] (sM) {$s^{(r)}_M$};

% expectation arrows
\draw[-stealth] ([yshift=-5pt]s1.south) -- +(0,-0.55)
node[below=3pt] {$\tr[\vrho\vA_1]$};

\draw[-stealth] ([yshift=-5pt]si.south) -- +(0,-0.55)
node[below=3pt] {$\tr[\vrho\vA_i]$};

\draw[-stealth] ([yshift=-5pt]sM.south) -- +(0,-0.55)
node[below=3pt] {$\tr[\vrho\vA_M]$};

% brace
\draw[decorate,decoration={brace,amplitude=8pt}]
(-0.8,-3.9) -- (-0.8,0.3)
node[midway,left=10pt]
{$\begin{array}{c}
5\frac{\ln(2M/\delta)}{\varepsilon^2}\\
\text{independent samples}
\end{array}$};

% correlated arrow
\draw[<->] ($(E11.north)+(0,0.7)$) -- ($(E1M.north)+(0,0.7)$)
node[midway,above] {correlated};

\draw[<->] ($(E1M.east)+(0.9,0)$) -- ($(ErM.east)+(0.9,0)$)
node[midway,above,sloped] {independent};

\end{tikzpicture}
\caption{Parallel measurements $\CE_i$ applied to multiple samples of $\vrho$ to estimate the expectation value of $\tr[\vrho \vA_i]$. We estimate $\tr[\vrho \vA_i]$ to precision $\varepsilon\le\frac12$ with failure probability $\leq\delta/M$ by averaging each outcome $s_i^{(k)}$ (obtained on the $k$-th sample); this requires a number of independent copies that scales merely with $\log(M)$. By the union bound this ensures an overall $\le \delta$ failure probability. Here, each $\CE_i$ for a given $\vA_i$ will further be decomposed into repeated detailed-balanced weak measurement channels $\CE_i(w)$, which are constructed in \autoref{lem:discMeasChannel}. 
}
\label{fig:parallel-measurements}
\end{figure}
\begin{lem}[Efficient detailed-balanced weak-measurement channels from Gibbs sampling]\label{lem:discMeasChannel}
    For any self-adjoint observable $\vA$ of operator norm at most $1$, Hamiltonian $\vH$, and inverse temperature $\beta\ge0$, there is a $\vrho_\beta$-detailed-balanced $w$-weak measurement channel with Kraus operators $\{\vK_+,\vK_0,\vK_-\}$ and a scalar $w=\Theta\Big(\frac{1}{\log^2(2+\beta\nrm{\vH})}\Big)$ such that $\|\vK_+^\dagg\vK_++\vK_-^\dagg\vK_-\|\le \frac{2}{\sqrt{3}}w$ and
    \begin{align*}
    \tr[\vrho_\beta\vK^{\dagger}_+\vK_+]-\tr[\vrho_\beta\vK^{\dagger}_-\vK_-] = w\tr[\vrho_\beta \vA].
    \end{align*}
    Furthermore, this channel can be implemented to diamond norm error $\veta$ using $\tCO(\beta)$ Hamiltonian simulation time and $\tCO(1)$ uses of an $a$-qubit block-encoding for the observable $\vA$ with the help of $\tCO(a+1)$ many ancillas.
\end{lem}

For measuring one specific observable, this bears no benefit compared to directly measuring $\vA$. However, having the averaged channel detailed-balanced will become crucial for estimating many observables. 

For our proof, recall the Bohr-frequency decomposition~\cite{chen2023ExactQGibbsSampler} 
\begin{align}
        \vA = \sum_{\nu\in B(\vH)} \vA_{\nu},\quad \text{where}\quad \vA_\nu&:=\sum_{E_2 - E_1 = \nu } \vP_{E_2} \vA \vP_{E_1} \quad \text{satisfies that} \quad (\vA_{\nu})^{\dagger} = (\vA^{\dagger})_{-\nu},\label{eqn:Aoperator}
\end{align}
where $\vP_{E}$ is the eigenspace projector with respect to $\vH$ for an energy $E\in \text{Spec}(\vH)$.

\begin{proof}
For any self-adjoint observable $\vA,$ consider
\begin{align*}
    \hat{\vA}_{f} &:= \int_{-\infty}^{\infty} f(t) e^{\ri \vH t}\vA e^{-\ri \vH t} \rd t\\&
     = \sum_{\nu} \vA_{\nu} g(\nu) \tag*{(setting $g(-\nu):=\hat{f}(\nu)=\int_{-\infty}^{\infty} f(t) e^{-\ri \nu t} \rd t$)}
\end{align*}
and construct a pair of ``polarizations''
\begin{align*}\\[-10mm]
    \vA_+ &:= \frac{1}{2}\vI + \frac{1}{2}\hat{\vA}_f, \text{ and} \\
    \vA_- &:= \frac{1}{2}\vI - \frac{1}{2}\hat{\vA}_f,
\end{align*}
which will be proportional to our ultimate Kraus operators. We would like to have a normalized $f$ function
\begin{align*}
    \int_{-\infty}^{\infty} \labs{f(t)} \rd t = 1,
\end{align*}
such that the function $g(\nu)=\hat{f}(-\nu)$ has the property (cf.~\cite[Theorem 8]{gilyen2024QGeneralGlauberMetropolis})
\begin{align*}
    g^*(\nu) &= g(-\nu) e^{-\beta \nu/2}
\end{align*}
so that crucially
\begin{align*}\\[-10mm]
e^{\beta \vH/2}\vA_+ e^{-\beta \vH/2} &= \vA_+^{\dagger}, \text{ and} \\
e^{\beta \vH/2}\vA_- e^{-\beta \vH/2} &= \vA_-^{\dagger},
\end{align*}
which ensures that the transition part of the final channel is detailed-balanced~\cite{ding2024GibbsSamplingViaKMS,gilyen2024QGeneralGlauberMetropolis}. 

A good choice is a shifted Gaussian
\begin{align*}
    g(\nu)= \hat{f}(-\nu) := \exp\left(- \frac{(\sigma \nu + \beta/\sigma)^2}{8}\right)
\end{align*}
%Verified in Mathematica:
%(Exp[-(σ ν + β/σ)^2/8]/(Exp[-(-σ ν + β/σ)^2/8]*Exp[-β ν/2]) // FullSimplify[#, Assumptions -> σ > 0 && β > 0 && Element[ν, Reals]] & ) == 0
so that, up to phase, $f$ is a Gaussian with standard deviation $\frac{\sigma}{2}$:
\begin{align*}
   f(t)=\frac{\sqrt{2 / \pi }}{\sigma} \exp\left(-\frac{2t^2 -\ri \beta t }{\sigma^2}\right).
\end{align*}
%Verified in Mathematica:
%1 == (-σ ν + β/σ)^2/8 + (FourierTransform[Sqrt[2/Pi]/σ Exp[-(2 t^2 - I β t)/σ^2],t, ν, FourierParameters -> {1, -1}] // Log) // FullSimplify[#,Assumptions -> σ > 0 && β > 0 && Element[ν, Reals]]&
We also have that 
$\norm{\vA_+},\norm{\vA_-} \le 1,$ 
since
$\norm{\hat{\vA}_{f}} \le \norm{\vA} \cdot \int_{-\infty}^{\infty} \labs{f(t)} \rd t = \norm{\vA} \leq 1.$ 
 %Verified in Mathematica:
 %1 == Integrate[Sqrt[2/Pi]/σ Exp[-2 t^2/σ^2], {t, -Infinity, Infinity}, Assumptions -> σ > 0]

The actual outcome probability will be determined by the (composite) observables
\begin{align*}
    \vA_+^{\dagger}\vA_+ =  \frac{1}{4}(\vI + \hat{\vA}_f)^{\dagger} (\vI + \hat{\vA}_f)\\
    \vA_-^{\dagger}\vA_- =  \frac{1}{4}(\vI - \hat{\vA}_f)^{\dagger} (\vI - \hat{\vA}_f).
\end{align*}
Adding the two, we obtain
\begin{align}\label{eq:KrausBound}
	\nrm{\vA_+^\dagg\vA_+\!+\vA_-^\dagg\vA_-} 
	= \bigg\lVert\frac{\vI}{2}+\!\frac{(\hat{\vA}_f)^{\dagger}\hat{\vA}_f}{2}\bigg\rVert 
	\le \frac{1+\lVert\hat{\vA}_f\rVert^2}{2}
	\le \frac{1+\lVert\vA\rVert^2}{2}
	\le 1,
\end{align}
while by subtracting them, we obtain
\begin{align*}
	\vA_+^{\dagger}\vA_+-\vA_-^{\dagger}\vA_- =  \frac{1}{2}(\hat{\vA}_f)^{\dagger} + \frac{1}{2} \hat{\vA}_f.
\end{align*}

Therefore, for any state that is diagonal in the energy basis, in particular the Gibbs state, we have that
\begin{align*}
      \tr\left[\vrho(\vA_+^{\dagger}\vA_+-\vA_-^{\dagger}\vA_-)\right] &= \frac{1}{2}\tr[\vrho\hat{\vA}_{f}+\vrho(\hat{\vA}_f)^{\dagger}]\\&
      = \frac{1}{2}\tr[\vrho\sum_{\nu} \vA_{\nu} g(\nu)+\vrho\big(\sum_{\nu} \vA_{\nu} g(\nu)\big)^{\dagger}] \\&
      = \frac{1}{2}\tr[\vrho\vA_{0} g(0)+\vrho(\vA_{0})^\dagger g^*\!(0)] \tag{Gibbs state is diagonal}\\&
      = \mathrm{Re}(g(0))\tr[\vrho\vA_0]\tag{$\vA=\vA^{\dagger}$ implies that $\vA_{0} = (\vA_{0})^{\dagger}$}\\&
      = \mathrm{Re}(g(0))\tr[\vrho\vA] \tag{Gibbs state is diagonal}\\&
      = \exp\left(- \frac{\beta^2}{8\sigma^2}\right)\tr[\vrho\vA].
\end{align*}

Lastly, for a normalizing constant $w=\Theta\left(\frac{1}{\log^2(2+\beta\nrm{\vH})}\right)$ we set
\begin{align*}
    \CT[\cdot] &:= w \exp\left(\frac{\beta^2}{8\sigma^2}\right) (\vA_+ [\cdot] \vA_+^{\dagger} +\vA_- [\cdot] \vA_-^{\dagger}) \\
    &=: \vK_+[\cdot]\vK_+^{\dagger}+ \vK_-[\cdot]\vK_-^{\dagger}
\end{align*}
such that due to \autoref{eq:KrausBound} if $\sigma\ge \beta$  we are guaranteed to have
\begin{align*}
    \nrm{\vK_+^\dagg\vK_+\!+\vK_-^\dagg\vK_-}
    =w \exp\left(\frac{\beta^2}{8\sigma^2}\!\right) \nrm{\vA_+^\dagg\vA_+\!+\vA_-^\dagg\vA_-} 
    \le w \exp\left(\frac{\beta^2}{8\sigma^2}\!\right) \le \frac{2w}{\sqrt{3}}.
\end{align*}
Since the CP map $\CT$ is itself detailed balanced, the general recipe of~\cite{gilyen2024QGeneralGlauberMetropolis} gives a suitable rejection term\footnote{Note that this way we get a valid measurement channel whenever $w\leq\exp\big(\frac{-\beta^2}{8\sigma^2}\big)$, however we only know how to efficiently implement the corresponding ``rejection term'' $\vK_0$ when $w=\CO\big(\log^{-2}(\beta\nrm{\vH})\big)$, see~\cite{gilyen2024QGeneralGlauberMetropolis}.} with Kraus operator $\vK_0=\sqrt{\sqrt{\vrho_\beta}(\vI-\CT^\dagger[\vI])\sqrt{\vrho_\beta}}\vrho_\beta^{-\frac12}$ such that total channel is trace-preserving ($\vK_0^{\dagger}\vK_0+\vK_+^{\dagger}\vK_++\vK_-^{\dagger}\vK_- = \vI$) and detailed balanced.

For our choice of $w=\Theta\left(\frac{1}{\log^2(2+\beta\nrm{\vH})}\right)$, if we set $\sigma=\Theta\big(\beta\log^{\frac{3}{2}}(1/\veta) \big)$ a block-encoding of the additional (subnormalised) Kraus operator $\vK_0/2$ can be implemented up to error $\bigO{\veta}$ by a quantum circuit using $\tCO(\beta)$ Hamiltonian simulation time and $\tCO(1)$ uses of an $a$-qubit block-encoding for the observable $\vA$ with the help of $\tCO(a+1)$ many ancillas~\cite[Corollary 3, Lemma 12]{gilyen2024QGeneralGlauberMetropolis}. Using robust oblivious amplitude amplification~\cite{cleve2016EffLindbladianSim,gilyen2018QSingValTransf} this enables implementing a Stinespring dilation of the entire $w$-weak measurement channel defined by the Kraus operators $\{\vK_+,\vK_0,\vK_-\}$ up to error $\veta$ (in diamond norm).
\end{proof}

The above lemma provides a detailed-balanced measurement channel for $w
\vA$. In order to estimate $\tr[ \vA\vrho]$, one could just consider the classical random variable $a$ associated with the outcome from the ancilla such that
    \begin{align}\label{eq:simpleSample}
        a =\begin{cases}
            \phantom{-}\frac{1}{w} \quad &\text{if outcome label is}\quad +,\\
            -\frac{1}{w} \quad &\text{if outcome label is}\quad -,\\
            \phantom{-}0 \quad &\text{else.}
            \end{cases}
    \end{align}
    Then, the random variable $a$ is an unbiased estimator  
    \begin{align*}
    \BE[a] = \frac{1}{w}\tr[\vK^{\dagger}_+\vK_+\vrho]-\frac{1}{w}\tr[\vK^{\dagger}_-\vK_-\vrho] = \tr[\vrho \vA] .
    \end{align*}
However, directly estimating the expectation $\BE[a]$ incurs an extra multiplicative loss of $w$ for the number of samples. 

\begin{rmk}
    The subnormalizing constant $w$ is unwanted, and intuitively speaking, means that when the measurement channel is applied we effectively only perform a meaningful measurement with probability $w$. Indeed, $\|\vK_+^\dagg\vK_++\vK_-^\dagg\vK_-\|= \bigO{w}$ means that the probability of seeing measurement label $+$ or $-$ is $\bigO{w}$ for any quantum state. However, we can avoid overheads in sample complexity by ``boosting'' the signal via repeating the same detailed-balanced measurement channel $\Omega(1/w)$ times. This effectively removes the adverse effect of the subnormalization by $w$, see \autoref{lem:measurementTailBound} and \autoref{cor:UnbiasedConcentration} in \autoref{apx:repeatTailBond} for a formal statement.
\end{rmk}

\subsection{The proof of our main result: \autoref{thm:main}}

The following lemma is instructive.
\begin{lem}[Estimating one observable through a $\vrho$-detailed-balanced measurement channel]\label{lem:oneObs}
    Consider an observable $\vA$ such that $\norm{\vA}\le 1$ and a Hamiltonian $\vH$.
    Let $\varepsilon,\delta\in(0,\frac12)$, then we can estimate $\tr[\vrho\vA]$ (without bias) within error $\varepsilon$ with success probability at least $1-\delta$ using
    \begin{align*}
    S &\geq (3.9+2\varepsilon)\frac{\ln(2/\delta)}{\varepsilon^2}\quad \text{samples of the full-rank state} \quad \vrho,
    \end{align*}
    by applying a (perfect, error-free) measurement channel as in \autoref{lem:discMeasChannel} repeated $\lceil\frac{2}{w}\rceil$ times per sample.
    Furthermore, when marginalizing all measurement outcomes, each copy of $\vrho$ remains statistically the same.  
\end{lem}

\begin{proof}[Proof of~\autoref{lem:oneObs}]
We will apply the detailed-balanced measurement for $k=\lceil\frac{2}{w}\rceil$ times on each copy of $\vrho$ 
\begin{align*}
    \CE^{k(1)}[&\vrho] \rightarrow (a_1^{(1)},a_2^{(1)},\ldots,a_k^{(1)}) \rightarrow s^{(1)}\\
    &\vdots\\
    \CE^{k(r)}[&\vrho] \rightarrow (a_1^{(r)},a_2^{(r)},\ldots,a_k^{(r)}) \rightarrow s^{(r)}
\end{align*}
where
\begin{align*}
    \CE^{k(i)} = \undersetbrace{k \quad \text{times}}{\CE^{(i)}\circ \cdots \circ \CE^{(i)}}.
\end{align*}
Since the measurement channel is detailed-balanced, after each application of the channel (ignoring previous measurement results, i.e., tracing out ancillas) the quantum state remains distributed as $\vrho$. Thus the random variable $a_j$ from \autoref{eq:simpleSample} associated to the $j$-th application of the measurement channel satisfies $\BE[a_j]=\tr[\vrho\vA]$.

Since in each application of the channel, the probability of seeing a non-zero ancilla label is at most $\frac{2}{\sqrt{3}}w$ for any prior measurement outcome history, an elementary exercise in probability (see \autoref{lem:measurementTailBound}) shows that taking the average of $\ge \frac{1}{w}$ subsequent measurements $a_j$ yields a sample mean $s=\frac{a_1+\ldots + a_k}{k}$ with bounded variance while by the linearity of expectation we have $\BE[s]=\tr[\vrho\vA]$. Then, using standard median-of-means estimation (see \autoref{lem:MOM}) we can obtain an estimate of $\BE[s]=\tr[\vrho\vA]$ with the desired precision and success probability using the independent measurement statistics $s^{(i)}$ from $r=\Theta\big(\frac{\log(1/\delta)}{\varepsilon^2}\big)$ copies of the mixed state $\vrho$.

Alternatively, if we take $k=\lceil\frac{2}{w}\rceil$ subsequent applications of the measurement channel, and average the independent measurement statistics $s^{(i)}$ obtained on $r\geq (3.9+2\varepsilon)\frac{\ln(2/\delta)}{\varepsilon^2}$ copies of $\vrho$, then the average is an unbiased estimator of $\BE[s^{(i)}]=\tr[\vrho\vA]$ that is within precision $\veps$ with probability at least $1-\delta$ as shown by Bernstein's concentration inequality (see \autoref{cor:UnbiasedConcentration} for the details of the argument).
\end{proof}

\begin{proof}[Proof of \autoref{thm:main}]
The idea is to simply apply the procedure of \autoref{lem:oneObs} to each of the $M$ observables subsequently, but with an elevated success probability $1-\frac{\delta}{2M}$. Then, by the union bound, we get an estimate of all expectation values that are correct simultaneously with probability at least $1-\frac{\delta}{2}$.
Applying the protocol of \autoref{lem:oneObs} to different subsequent observables is possible, because we only apply detailed-balanced quantum channels on each copy of the Gibbs state, therefore each time we apply \autoref{lem:oneObs} its assumption is satisfied, we have $r$ independent copies of the Gibbs state (when ignoring previous measurement outcomes).\footnote{Considering a single Gibbs state, the sequential statistics of the $M$ different observables $(s_{1}^{(1)},\ldots,s_{M}^{(1)}),$ may be highly \textit{correlated}, however this does not matter, as we are using a union bound regarding the error probability of different observables.}

We can implement each measurement channel to precision $\frac{\delta}{2 M r k}=\Omega\big(\poly\left(\frac{w\delta\veps}{M}\right)\big)$ in diamond norm according to \autoref{lem:discMeasChannel}. Then, by the composition properties of the diamond norm, the entire protocol is implemented to diamond norm error $\frac{\delta}{2}$, implying that the measurement statistics cannot change by more than $\frac{\delta}{2}$ compared to the protocol using the ideal channels, so we succeed with probability at least $1-\delta$ even using the implemented slightly noisy channels.
\end{proof}

\subsection{Continuous-time approach}

While the above exploits discrete-time Gibbs samplers, an analogous result can be achieved with continuous-time Gibbs sampling, which is very similar to the discrete-time case, except that the measurement outcomes are described by jump occurrences.

The main idea is to apply a sequential $\vrho$-detailed-balanced continuous-time measurement for $t=\bigOt{1}$ and $\vA_1,\cdots, \vA_M,$ and repeat $r$ times on $r$ copies of $\vrho$
\begin{align*}
    \CE^{(1)}_{M}\cdots\CE^{(1)}_1[&\vrho]\quad \text{such that}\quad \exp(\CL^{(1)}_{M})=\cdots=\exp(\CL^{(1)}_{1})[\vrho]=\vrho \\
    &\vdots\\
    \CE^{(r)}_{M}\cdots\CE^{(r)}_1[&\vrho] \quad \text{such that}\quad \exp(\CL^{(r)}_{M})=\cdots=\exp(\CL^{(r)}_{1})[\vrho]=\vrho.
\end{align*}
The rest of the argument is completely analogous to the above discrete case and~\autoref{cor:continuous_time_one}. 

For completeness, we provide explicit definitions of the required continuous-time measurement and provide a high-level proof sketch.  

\begin{defn}[Detailed-balanced continuous-time measurements]\label{def:Meas}
    Given a Hermitian observable $\vA$ and a full-rank quantum state $\vrho$, we say a collection of Lindblad operators $\{\vL_j\}$ and a Hermitian ``coherent term'' $\vC$, forms a $\vrho$-detailed-balanced continuous-time measurement for $\vA$ if $\|\sum_j\vL_j^{\dagger}\vL_j\|=\bigO{1}$, and
    \begin{itemize}
    \item (detailed-balanced on average) When marginalizing over the observed jump labels $j$, the resulting Lindbladian $\CL[\cdot] = -\ri [\vC,\cdot] + \sum_{j\in [\ell]} \vL_j[\cdot]\vL^{\dagger}_j- \frac{1}{2}\{\vL_j^{\dagger}\vL_j,\cdot\}$ satisfies (KMS) detailed balance with respect to $\vrho$.
	\item (informative outcomes) A bounded linear combination of the rates $\braket{\vL_j^{\dagger}\vL_j}_{\vrho}$ provides an unbiased estimator of $\braket{\vA}_{\vrho}$, i.e., $\braket{\vA}_{\vrho}=\sum_{j\in [\ell]} \lambda_j \braket{\vL_j^{\dagger}\vL_j}_{\vrho}$ for some known coefficients $\lambda_j\in[-1,1]$.
\end{itemize}
\end{defn}

The continuous-time Lindbladian evolution can be described using its Duhamel-Dyson series as follows:
\begin{align*}
	e^{t\CL}[\cdot] &= \sum_{k=0}^{\infty}\sum_{(j_1,\ldots,j_k)\in[\ell]^k} \iiint\limits_{\kern2mm t\ge t_k > \cdots> t_1 >0} \vK_{j_k,t_k,\cdots,j_1,t_1}[\cdot] \vK_{j_k,t_k,\cdots,j_1,t_1}^{\dagger}\rd t_1\cdots \rd t_k, \text{ where }\\
	\vK_{j_k,t_k,\cdots,j_2,t_2,j_1,t_1}&=e^{(t-t_{k})\vH_{\mathrm{eff}}}\vL_{j_k}e^{(t_k-t_{k-1})\vH_{\mathrm{eff}}}\cdot\ldots\cdot\vL_{j_2}e^{(t_2-t_1)\vH_{\mathrm{eff}}}\vL_{j_1}e^{t_1\vH_{\mathrm{eff}}} \text{ and }
	\vH_{\mathrm{eff}}=-\ri \vC -\frac{1}{2}\sum_{j\in [\ell]} \vL_j^{\dagger}\vL_j.
\end{align*}

It will also be useful to distinguish between the measurement protocol $\CE^{t}[\cdot]$ that records the trajectories $(j_k,t_k,\cdots,j_2,t_2,j_1,t_1)$, and the averaged channel $e^{t \CL}=\tr_{traj}[\CE^{t}]$. Applying $\CE^{t}[\cdot]$ on $\vrho$ with a Lindbladian as in the above definition, we can get an unbiased estimator for the desired observable
\begin{align}
\braket{\vA}_{\vrho}=\frac{1}{t}\sum_{j\in [\ell]} \lambda_j \cdot\BE[\text{number of observed jumps }j]    .
\end{align}
Unlike for the discrete-time version, the trajectory unraveling here uses a continuum of Kraus operators~\cite{li2022SimMarkOpen}. In an actual implementation, we need to represent the trajectories finitely, so we introduce coarse jump times $\bar{t}$ that are discrete multiples of a fixed~$\tau$: 
\begin{align*}
    \bar{\CE}_\tau^{t}[\cdot] = \sum_{k=0}^{\!\lfloor t/\tau\rfloor\!}
    \sum_{\underset{t\ge \bar{t}_k>\cdots > \bar{t}_1> 0}{\,(j_1,\ldots,j_k)\in[\ell]^k}\,}\!\!
    \left(\iiint\limits_{\kern2mm t_i \in(\bar{t}_i-\tau,\bar{t}_i]} \!\!\!\!\vK_{j_k,t_k,\cdots,j_1,t_1}[\cdot] \vK_{j_k,t_k,\cdots,j_1,t_1}^{\dagger}\rd t_1\cdots \rd t_k \! \right)\!\otimes
    \ket{j_k,\bar{t}_k,\cdots,j_1,\bar{t}_1}\bra{j_k,\bar{t}_k,\cdots,j_1,\bar{t}_1}.
\end{align*}
This is done by ``rounding'' the event time-labels to the next discrete time step and ignoring cases when multiple jumps occur in the same discrete $\tau$ time-window. This introduces only a small (subnormalising) discretisation error $\|e^{t \CL}-\tr_{traj}[\bar{\CE}_\tau^{t}]\|_\diamond=\bigO{t\tau}$ proportional to $\tau$, which can be ignored when the discretisation is fine enough. 

\begin{lem}[Detailed-balanced continuous-time measurement from Gibbs sampling]\label{lem:contMeasImp}
For any self-adjoint observable $\vA,$ Hamiltonian $\vH$, and inverse temperature $\beta\ge 0$, there is a $\vrho_\beta$-detailed-balanced continuous-time measurement $\{\vL_i\}$ (with appropriate coherent term $\vC$) as in \autoref{def:Meas} such that 
    \begin{align*}
    \tr[\vL^{\dagger}_1\vL_1\vrho_\beta]-\tr[\vL^{\dagger}_2\vL_2\vrho_\beta] = \tr[ \vA\vrho_\beta].
    \end{align*}
    Furthermore, a quantum channel $\veta$-approximating $\bar{\CE}_\veta^{t}$ (in the completely bounded one-to-one norm $\|\cdot\|_\diamond$) can be implemented using $\tCO(t\beta)$ Hamiltonian simulation time and $\tCO(1)$ block-encodings for the observable $\vA,$ and $\tCO(1)$ many ancillas, assuming (for simplicity) that $\vA$ is provided with an $\bigOt{1}$-ancilla block-encoding.
\end{lem}
\begin{proof}
Consider the Lindbladian with Lindblad operators $\vL_1\!:=\exp\left(\frac{\beta^2}{16\sigma^2}\!\right)\vA_+,\vL_2\!:=\exp\left(\frac{\beta^2}{16\sigma^2}\!\right)\vA_-$ as in~\autoref{lem:discMeasChannel} and invoke high-precision Lindbladian simulation algorithms~\cite{cleve2016EffLindbladianSim,chen2023QThermalStatePrep} to approximate the above $\bar{\CE}_\veta^{t}$.

As shown by~\cite{chen2023ExactQGibbsSampler,gilyen2024QGeneralGlauberMetropolis} the coherent term $\vC=\vC_1+\vC_2$ can be chosen and efficiently implemented via linear combination of unitaries (LCU) according to the formula
\begin{align*}
    \vC_j=\int_{\!-\infty}^\infty \frac{1}{\sinh(2 \pi t)}\left(e^{\ri \beta\vH t}\vL^{\dagger}_j\vL_je^{-\ri\beta \vH t}-\vL^{\dagger}_j\vL_j\right)\rd t.\tag*{\qedhere}
\end{align*}
\end{proof}

Most Lindbladian simulation algorithms~\cite{cleve2016EffLindbladianSim,chen2023QThermalStatePrep} discretize the time-evolution by effectively iterating $t/\tau$ times an approximate implementation of a tiny discrete time-step $\tau$
\begin{align*}
\bar{\CE}_\tau^{\tau}[\cdot]&=\vL_0[\cdot]\vL_0^\dagger\otimes\ketbra{0}{0} + \tau\sum_{j\in [\ell]} \vL_j[\cdot]\vL^{\dagger}_j\otimes\ketbra{j}{j}+\bigO{\tau^2}, \text{ where }\\
\vL_0&=\vI-\tau\sum_{j\in [\ell]}\bigg(\frac{\vL_j^{\dagger}\vL_j}{2}+\ri \vC_j\bigg)=\sqrt{\sqrt{\vrho_\beta}\Big(\vI-\tau\sum_{j\in [\ell]}\vL_j^{\dagger}\vL_j\Big)\sqrt{\vrho_\beta}}\vrho_\beta^{-\frac12}+\CO(\tau^2),
\end{align*}
making the approximation that only one jump is allowed in each discrete $\tau$ time-window. So we can think about the discrete time-steps of the resulting dynamics as a special case of the discrete-time measurement channel from \autoref{lem:discMeasChannel}, but with a tiny $w=\tau$ subnormalization. Since our proven sample complexity is independent of $w$, the previous analysis is still applicable, noting that in these implementation methods, the ancillae store the discrete trajectory labels in a potentially compressed form~\cite{cleve2016EffLindbladianSim,chen2023QThermalStatePrep}.

\begin{cor}[Estimating one observable through detailed-balanced continuous-time measurements]\label{cor:continuous_time_one}
    Consider an observable $\vA$ such that $\norm{\vA}\le 1$. Then, we can estimate $\tr[\vrho\vA]$ to error $\varepsilon$ with failure probability $\delta$ using
    \begin{align*}
    	S &= \CO\left(\frac{\log(1/\delta)}{\varepsilon	^2}\right)\quad \text{samples of the full-rank state} \quad \vrho,
    \end{align*}
    and a total of $\CO(1)$ Lindbladian simulation time per sample with a Lindbladian as in \autoref{lem:contMeasImp}.
    Moreover, if the Lindbladian simulation is executed perfectly, and we marginalize over all trajectory labels, then each copy of $\vrho$ remains statistically the same due to detailed balance.

 \end{cor}
 As in the proof of \autoref{thm:main}, we can use the above result subsequently to all $M$ measurement operators with elevated success probability $1-\frac{\delta}{2M}$ to recover all the $M$ expectation values to precision $\veps$ with success probability at least $1-\delta$. Finally, applying approximate Lindbladian simulation with imprecision $\ll\frac{\varepsilon}{M}$ ensures a bounded perturbation $\ll\varepsilon$ in each copy of $\vrho$ throughout the process.

\subsection*{Acknowledgements}

The authors developed the key ideas and wrote the manuscript. LLMs were used only to aid in completing details of the standard technical proofs of~\autoref{cor:UnbiasedConcentration}, \autoref{lem:measurementTailBound}, and \autoref{lem:stringSampLower} in the Appendices, and to provide feedback on the writing. CFC is supported by a Simons-CIQC postdoctoral fellowship through NSF QLCI Grant No.~2016245. AG is supported by the Lendület ``Momentum'' program of the Hungarian Academy of Sciences under grant agreement no.~LP2025-8/2025.
The authors are grateful for insightful discussions with Scott Aaronson, Anurag Anshu, Richard Allen, Ryan Babbush, Sitan Chen, Soonwon Choi, Hsin-Yuan Huang, William Huggins, Jiaqing Jiang, Robbie King, Daniel Mark, Ankur Moitra, Anand Natarajan, Nicholas Rubin, Matteo Votto, Alec White, and John Wright. 

\addcontentsline{toc}{section}{References}
\bibliographystyle{alphaUrlePrint.bst}
\bibliography{ref,qc_gily}

\addtocontents{toc}{\protect\let\protect\numberline\protect\appendixnumberline}
\appendix

\section{Concentration of the average of subsequent measurements}\label{apx:repeatTailBond}

In this appendix we first recall some useful concentration inequalities, which we then use for showing that $\Theta(\frac{1}{w})$ repetitions of our $w$-weak measurements statistically behave essentially as good as a single normal (full-strength) measurement, without requiring extra copies of the Gibbs state.

\begin{prop}[Bernstein's Inequality~{\cite[Corollary 2.11]{boucheron2013ConcIneqThIndep}}]\label{prop:BernsteinsIneq}
	Let $X_{i}\colon i\in [k]$ be independent real-valued random variables, and $c,v > 0$ positive numbers such that 
	$\frac{1}{k}\sum_{i=1}^k \mathbb{E}[X_i^2]\leq v$
	and for all integers $q\geq 3\colon$
	\begin{align*}
		\frac{1}{k}\sum_{i=1}^k \mathbb{E}[(X_i)_+^q]\leq\frac{q!}{2}v c^{q-2},
	\end{align*}
	where $x_+=\max(x,0)$ denotes the positive part of a number.
	Let
	$S_k=\frac{1}{k}\sum_{i=1}^k X_{i}-\mathbb{E}[X_i],$
	then for all $\varepsilon>0$
	\begin{align*}
		\Pr[S_k\geq \varepsilon]
		\leq\exp\left(-\frac{\varepsilon^2 k}{2(v+c\varepsilon)}\right).
	\end{align*}	
\end{prop}
	
\begin{prop}[Median of means{~\cite{Jerrum1986RandomGO}}]\label{lem:MOM}
    Let $X$ be a random variable with variance $\sigma^2.$ Then, $K$
    independent sample means of size $N = 34\sigma^2/ \varepsilon^2$ suffice to construct a median of means estimator 
    $\hat{\mu}(N, K)$ such that
    \begin{align*}
    \Pr [\labs{ \hat{\mu}(N, K) - \BE[X]} \ge \varepsilon ] \le 2e^{-K/2}.    
    \end{align*}
\end{prop}

\begin{lem}\label{lem:measurementTailBound}
	Let $w\in(0,1]$, and $X_0,X_1,X_2,\dots$ be (not necessarily independent) random variables taking values in $\big\{-\frac1w,\,0,\,\frac1w\big\}$, 
	such that almost surely
	\[
	\Pr[X_{i}\neq 0 \mid X_0,\dots,X_{i-1}]\le \frac{2}{\sqrt{3}} w.
	\]
	Then the sample mean
	$S_k=\frac{X_{0}+X_{1}+\cdots+X_{k-1}}{k}$ satisfies $\BE[S_k^2]\leq\frac{4}{3}+\frac{2}{\sqrt{3}kw}$, and for all $q>0$ we have
	\[
	\BE[|S_k|^q]\le\BE\bigg[\bigg(\frac{\mathrm{Bin}\big(k,\min(1,\frac{2}{\sqrt{3}}w)\big)}{kw}\bigg)^{\!\!q\,}\bigg].
	\]
\end{lem}

\begin{proof}
	Since $X_i\in\big\{-\frac1w,\,0,\,\frac1w\big\}$, for $Z_i:= \mathbf{1}_{\{X_i\neq 0\}}$ we have
	\begin{align}\label{eq:measureConcIndicators}
		|S_k|
		\le \frac1k\sum_{i=0}^{k-1} |X_i|
		= \frac{1}{wk}\sum_{i=0}^{k-1} Z_i.
	\end{align}
	
	The main idea is that we construct a coupling between the binary random variables $Z_i=\mathbf{1}_{\{X_i\neq 0\}}$ and $k$ independent binary Bernoulli random variables $B_i$ such that $w':=\min(1,\frac{2}{\sqrt{3}}w)=\BE[B_i]$. We realize the random variable $B_i$ using i.i.d.\ uniform random variables $U_i$ over $[0,1]$, by simply setting $B_i=\mathbf{1}_{\{U_i\leq w'\}}$.
	
	We define the random variable $Z'_i$ recursively as follows. Let $Z'_0:=\mathbf{1}_{\{U_0\leq \Pr[Z_0=1]\}}$ and let
	\[
	Z'_{i}:=\mathbf{1}_{\{U_{i}\leq \Pr[Z_{i}=1|(Z_0,Z_1,\ldots,Z_{i-1})=b]\}} \text{ if }(Z'_0,Z'_1,\ldots,Z'_{i-1})=b\in\{0,1\}^i.
	\]
	One can easily see that the law of $(Z_i)_{i=0}^{k-1}$ and $(Z'_i)_{i=0}^{k-1}$ coincide, since
	\begin{align*}     
		\Pr[(Z'_0,\ldots,Z'_{k-1})=b]
		&=\prod_{\ell=0}^{k-1} \Pr[Z'_\ell=b_\ell|(Z'_0,\ldots,Z'_{\ell-1})=(b_0,\ldots,b_{\ell-1})]\\
		&=\prod_{\ell=0}^{k-1} \Pr[Z_\ell=b_\ell|(Z_0,\ldots,Z_{\ell-1})=(b_0,\ldots,b_{\ell-1})]
		=\Pr[(Z_0,\ldots,Z_{k-1})=b].
	\end{align*}   	
	On the other hand, the coupling is constructed in such a way that $Z'_i\leq B_i$. Therefore, we have that 
	\begin{align}\label{eq:Coupling}
		\sum_{i=0}^{k-1} Z'_i \le \sum_{i=0}^{k-1} B_i\sim \mathrm{Bin}(k,w').
	\end{align}
	
	By \autoref{eq:measureConcIndicators}-\ref{eq:Coupling} and the above coupling we have for all $q > 0$ 
	\begin{align*}
		\BE[|S_k|^q] \leq \frac{1}{(kw)^q}\BE[(\mathrm{Bin}(k,w'))^q].
	\end{align*} 
	In particular, for $q=2$ this implies
	\begin{align*}
		\BE[|S_k|^2]
		\leq \frac{\BE[\mathrm{Bin}(k,w')]^2+\mathrm{Var}[\mathrm{Bin}(k,w')]}{(kw)^2}
		\leq \frac{4}{3}+\frac{w'(1\!-\!w')}{kw^2}
		\le\frac{4}{3}+\frac{2}{\sqrt{3}kw}.\tag*{\qedhere}
	\end{align*} 
\end{proof}

\begin{cor}\label{cor:UnbiasedConcentration}
	Let $w\in(0,1]$, and $X^{(i)}_0,X^{(i)}_1,X^{(i)}_2,\dots$ be as in \autoref{lem:measurementTailBound}, then 
	the sample mean
	$S^{(i)}_k=\frac{X^{(i)}_{0}+X^{(i)}_{1}+\cdots+X^{(i)}_{k-1}}{k}$ satisfies $|S^{(i)}_k|\le \frac{1}{w}$ almost surely. Let $\varepsilon,\delta\in (0,\frac12)$, and suppose that $S^{(i)}_k$ are mutually independent for $i\in [r]$, then for $k\geq\frac{2}{w}$ we have that $O_r:=\frac{S^{(1)}_{k}+S^{(2)}_{k}+\cdots+S^{(r)}_{k}}{r}$ satisfies 
	\begin{align}\label{eq:averageConc}
		\Pr[|O_r-\BE[O_r]|>\varepsilon]\leq \delta \text{ for all }r\ge(3.9+2\varepsilon)\frac{\ln(2/\delta)}{\varepsilon^2}.
	\end{align}
\end{cor}
\begin{proof}
	\autoref{eq:averageConc} follows directly from Bernstein's inequality~(\autoref{prop:BernsteinsIneq}) setting $X_i:=S_k^{(i)}$, $v=1.95$, and $c=1$. Indeed, by \autoref{lem:measurementTailBound} we have that 
	\begin{align*}
		\BE[S_k^2]\leq\frac{4}{3}+\frac{2}{\sqrt{3}kw}\leq\frac{4}{3}+\frac{1}{\sqrt{3}}< 1.95=v.
	\end{align*}
	Let $w':=\min(1,\frac{2}{\sqrt{3}}w)$; to verify the validity of our choice of $c=1$ we now show that for all $q\ge 3$ integer
\begin{align}
	\BE[|S_k|^q]
	&\le \BE\bigg[\bigg(\frac{\mathrm{Bin}\big(k,w'\big)}{kw}\bigg)^{\!q\,}\bigg] \tag{by \autoref{lem:measurementTailBound}}\\
	&\le \Big(\frac{q}{kw\ln(1+q/(kw'))}\Big)^{\!q\,}\tag{Ahle's raw moment bound~\cite{ahle2022SharpBoundsMomentsBinomPoisson}}\\
	&\le \Big(\frac{q}{2\ln(1+(qw/w')/2)}\Big)^{\!q\,}\tag{$kw\ge 2$ and $\forall a>0\colon x\mapsto x\ln(1+a/x)$ is monotone increasing}\\
	&\le \Big(\frac{q}{2\ln\big(1+\frac{\sqrt3}{4}q\big)}\Big)^{\!q\,} \tag{$\frac{w}{w'}\ge \frac{\sqrt{3}}{2}$}\\
	&\le 0.975 q! = \frac{q!}{2}vc^{q-2},\label{eq:momLastStep}
\end{align}
where the last inequality for $q\ge 7$ follows from
\begin{align}
	\frac{\Big(\frac{q}{2\ln\big(1+\frac{\sqrt3}{4}q\big)}\Big)^{\!q\,}}{q!}
	&< \Big(\frac{e}{2\ln\big(1+\frac{\sqrt3}{4}q\big)}\Big)^{\!q\,} \tag{since $q!> (q/e)^{q}$}\\
	&\le \Big(\frac{e}{2\ln\big(1+7\sqrt{3}/4\big)}\Big)^{\!q\,}  \tag{since $q\geq 7$}\\
	&\le \Big(\frac{e}{2\ln\big(1+7\sqrt{3}/4\big)}\Big)^{\!7} 
	< 0.975.\nonumber
% Checked by Mathematica:
%  In: (E/2/Log[1 + 7*Sqrt[3]/4])^7 // N
% Out: 0.837399
\end{align}

\begin{table}[ht]
	\centering
	\newcommand{\Mratio}[1]{\fpeval{round((#1/(2*ln(1+sqrt(3)/4*#1)))^#1/fact(#1), 5)}}
	\begin{tabular}{@{}ccccc@{}}
		\toprule
		$q$ & 3 & 4 & 5 & 6 \\
		\midrule
		\rule[-2.6ex]{0pt}{6.6ex}$\dfrac{1}{q!}\bigg(\dfrac{q}{2\ln\big(1+\frac{\sqrt3}{4}q\big)}\bigg)^{\!q\,}$
		& $\Mratio{3}$ & $\Mratio{4}$ & $\Mratio{5}$ & $\Mratio{6}$ \\
		\bottomrule
	\end{tabular}
	\caption{Direct verification of the inequality \eqref{eq:momLastStep} for $q=3,4,5,6$.}
	\label{tab:ahle}
\end{table}

Finally, for $3\leq q \leq 6$ we can verify the inequality \eqref{eq:momLastStep} by direct computation as shown in \autoref{tab:ahle}.
\end{proof}

\section{Sample optimality of our protocol in a black-box model}\label{apx:sampLower}

In this appendix we present a detailed proof of \autoref{thm:tightSample}. We begin by proving a series of technical lemmas, building on which we complete the proof at the end.

\subsection{Our zero-temperature Gibbs states are no harder than low-temperature ones}\label{apx:lowTempToZeroTemp}

We use the following lemma to show that the zero-temperature state of our classical binary Hamiltonian on $n$-bits is close to a reasonable temperature Gibbs state, so that the simplification of working with a zero-temperature Gibbs state is essentially without loss of generality:
\begin{lem}\label{lem:tempTVBound}
	For every $\beta\geq \ln(2)n+r$ and $f:\{0,1\}^n\to\{0,1\}$ such that $|\{b:f(b)=0\}| \geq 1$ we have 
	\[\|\rho_\beta - \rho_\infty\|_{\mathrm{TV}}\leq e^{-r}.\]
\end{lem}
\begin{proof}
	At inverse temperature $\beta$ the Gibbs distribution is $\rho_\beta(x) \propto e^{-\beta f(x)}$.
	Let us define $k := |\{b \colon f(b)=0\}| \geq 1$. 
	Since $f$ is Boolean, zeros receive weight~$1$ and ones receive weight~$e^{-\beta}$, so that
	$ Z_\beta \;=\; k + (2^n - k)\,e^{-\beta}$ and 
	\begin{equation*}\label{eq:test}
		\rho_\beta = \frac{\chi[f(b)=0]+e^{-\beta}\chi[f(b)=1]}{Z_\beta}.
	\end{equation*}
	
	At $\beta = \infty$ the Gibbs state reduces to the uniform distribution over the $k$ zeros:
	$\rho_\infty = \frac{\chi[f(x)=0]}{k}$, assigning more weight to preimages of $0$ and less (actually zero) to preimages of $1$ compared to the finite temperature $\rho_\beta$. Therefore, using that $e^{-\beta}\leq 2^{-n}\cdot e^{-r}$, it is easy to see that 
	
	\begin{align*}
		\|\rho_\beta - \rho_\infty\|_{\mathrm{TV}}
		= \frac{1}{2}\sum_{b\in\{0,1\}^n} \!\!|\rho_\beta(b) - \rho_\infty(b)|
		=\! \sum_{b\colon f(b)=1} \!\!\rho_\beta(b)-0
		= (2^n\! - k)\frac{e^{-\beta}}{Z_\beta} 
		\leq (2^n\! - k)\frac{2^{-n}\cdot e^{-r}}{k}
		\leq e^{-r}.\!\tag*{\qedhere}
	\end{align*}
\end{proof}

Note that the above readily applies to Hamiltonians with $\nrm{\beta \vH}=\bigO{n}$, which is the natural scale for local Hamiltonians, constraint satisfaction problems, etc.

\subsection{Learning distributions given by uniform sampling over a subset}\label{apx:stringSampLower}

A key step in our proof is binning: we reduce a distribution on $2^n$ elements to one on $m\ll 2^n$ elements. Moreover, in our proof, the original distribution is uniform over a subset $I$ of the $2^n$ elements. The following lemma shows that if $K=|I|$ is large enough, then such induced distributions over $m$ bins are essentially no easier to work with than arbitrary distributions that are natively given on $m$ elements. In order to handle the discretization of arbitrary distributions we introduce a form of rounding that produces distributions whose probabilities are multiples of $1/K$. Consider the ``rounding function'' $q = r(p)$ that carefully rounds the probabilities of $p\in\Delta_m$ to multiples of $1/K$ while keeping their sum equal to~$1$, where $q_i$ is defined recursively as
\[
q_1 = \frac{\lfloor p_1 K +\frac12 \rfloor}{K},\qquad
q_1 + q_2 = \frac{\lfloor (p_1+p_2)\, K +\frac12 \rfloor}{K},\qquad \ldots
\]

It is easy to see that $|q_i-p_i|\leq 1/K$, and therefore
$d_{\mathrm{TV}}(p,r(p)) \le \frac{m}{2K}$. 

Also for a $P\subseteq \Delta_m$ let us define $r(P):=\{r(p)\colon p\in P\}$.

\begin{restatable}{lem}{stringSampLower}\label{lem:stringSampLower}
	Let $f\colon\Delta_m \to \mathbb{R}^k$ be $L$-Lipschitz with respect to the metric $D$ and total variation distance:
	\[
	D(f(p),f(q)) \le L \cdot d_{\mathrm{TV}}(p,q)
	\qquad\text{for all } p,q\in\Delta_m.
	\]
	Assume that any algorithm which, given i.i.d.\ samples from an unknown distribution $p\in P\subseteq\Delta_m$, outputs a value $v\in\BR^k$y such that
	\[
	D(f(p),v) \le 2\veps
	\]
	with probability at least $1-3\delta$, must use at least $S+1$ samples from $p$.
	
	Let
	\[
	K \;\ge\; \max\left\{\frac{Lm}{2\veps},\frac{Sm}{2\delta},\frac{S^2}{2\delta}\right\}, \qquad \text{and} \qquad n\geq \log_2(m K). 
	\]
    Given an unknown distribution $q\in r(P)$ on $m$ elements, partition $\{0,1\}^n$ into $m$ bins $B_1,\dots,B_m$ of sizes differing by at most $1$. 
    Let $I$ be a uniformly random (unknown) subset $I\subseteq\{0,1\}^n$ with $|I|= K$, so that
	$\frac{|I\cap B_i|}{|I|} \;=\; q_i$.
	Then any algorithm that outputs $v\in\BR^k$ such that
	\[
	D(f(q),v) \le \veps
	\]
	with probability at least $1-\delta$, must use at least $S+1$ i.i.d.\ uniform random samples from elements of $I$.
\end{restatable}
\begin{proof}	
	We first show that even for the rounded distributions $r(P)$, estimating $f$ to accuracy $\veps$ with confidence $1-2\delta$ requires at least $S+1$ samples. The key observation is that with our choice of $K$ for every $p\in \Delta_m$ we have
	\[
	d_{\mathrm{TV}}(p,r(p)) \le \frac{m}{2K}
	\le \min\left\{\frac{\delta}{S},\frac{\veps}{L}\right\}.
	\]
	
	Towards a contradiction suppose that there exists an algorithm $\CA$ using at most $S$ samples such that for every $q\in r(P)$,
	\[
	\Pr\left[D(\CA(X_1,\dots,X_{S}),f(q))\le \veps\right]\ge 1-2\delta,
	\]
	where $X_1,\dots,X_{S}$ are i.i.d.\ from $q$.
	
	Now let $p\in P$ be arbitrary, and set $q:=r(p)$. Run $\CA$ on $S$ i.i.d.\ samples from $p$. By coupling the two one-sample distributions and taking a union bound over the $S$ samples,
	\[
	d_{\mathrm{TV}}\!\left(p^{\times S},q^{\times S}\right)
	\le S\,d_{\mathrm{TV}}(p,q)
	\le \delta.
	\]
	Hence the success probability of $\CA$ for estimating $f(q)$ with samples from $p$ is at least
	\[
	(1-2\delta)-\delta = 1-3\delta.
	\]
	Furthermore, whenever $\CA$'s output $v$ satisfies $D(v,f(q))\le \veps$, by the triangle inequality we have
	\[
	D(v,f(p))
	\le D(v,f(q)) + D(f(q),f(p))
	\le \veps + L\cdot d_{\mathrm{TV}}(p,q)
	\le 2\veps.
	\]
	So we would obtain an $S$-sample algorithm that estimates $f(p)$ to accuracy $2\veps$ with confidence $1-3\delta$, contradicting the assumed lower bound. This proves the claim about $r(P)$.
	
	We now pass to the subset-element-sampling model. Fix $q\in r(P)$. Since each $Kq_i$ is an integer, we can realize $q$ by a subset $I\subseteq\{0,1\}^n$ of size $K$ satisfying
	\[
	|I\cap B_i| = Kq_i
	\qquad\text{for all } i\in[m].
	\]
	Assume, toward contradiction, that there is an algorithm $\CA$ which uses at most $S$ samples from uniformly random elements of $I$, and outputs an $\veps$-accurate estimate of $f(q)$ with probability at least $1-\delta$ when $\CA$ is run for a uniformly random such $I$.
	
	We compare two transcript distributions of length $S$.
	
	\medskip
	\noindent
	\textbf{(1) True subset-element-sampling transcript.}
	Sample $Y_1,\dots,Y_S$ i.i.d.\ uniformly from the $K$ elements of $I$.
	
	\medskip
	\noindent
	\textbf{(2) Simulated transcript from sample access to $q$.}
	Sample $J_1,\dots,J_S$ i.i.d.\ from $q$. Each time $J_t=i$, we output a fresh uniformly random string from bin $B_i$, never reusing a previously output string from that bin.
	
	\medskip
	
	The second procedure only needs sample access to $q$. Moreover, we can couple the two transcripts such that they only differ if the subset-element-sampling returns some element of $I$ twice, i.e., $Y_i=Y_j$ for some $i\neq j$. Indeed, consider a sequence of distinct elements $(e_1,e_2,\ldots,e_S)$ and let $(b_1,b_2,\ldots,b_S)$ denote the corresponding bins, then considering that $I$ is uniformly random we have that
	\begin{align*}
		\Pr[(Y_1,\ldots,Y_S)\!=\!(e_1,\ldots,e_S)]&
		=\Pr[(Y_1,\ldots,Y_S)\!=\!(e_1,\ldots,e_S)| Y_1\!\in\! B_{b_1},\ldots,Y_S\!\in\! B_{b_S}]
		\cdot \Pr[ Y_1\!\in\! B_{b_1},\ldots,Y_S\!\in\! B_{b_S}]\\&
		\leq \Pr[ Y_1\!\in\! B_{b_1},\ldots,Y_S\!\in\! B_{b_S}]\cdot \prod_{i=1}^{S}\frac{1}{|B_{b_i}|-|\{j<i\colon b_j=b_i\}|}\\&
		= \Pr[J_1=b_1,\ldots,J_S=b_S]\cdot \prod_{i=1}^{S}\frac{1}{|B_{b_i}|-|\{j<i\colon b_j=b_i\}|}\\&
		=\Pr[(\text{simulated transcript})\!=\!(e_1,\ldots,e_S)].
	\end{align*}
	Thus
	\[
	d_{\mathrm{TV}}\!\left((Y_1,\dots,Y_S),(\text{simulated transcript})\right)
	\le \Pr[\exists\, s<t\leq S:\ Y_s=Y_t].
	\]
	By the union bound,
	\[
	\Pr[\exists\, s<t\leq S:\ Y_s=Y_t]
	\le \binom{S}{2}\frac1K
	\le \frac{S^2}{2K}
	\le \delta.
	\]
	Therefore, the simulated procedure also succeeds with probability at least
	\[
	(1-\delta)-\delta = 1-2\delta.
	\]
	
	But this yields an algorithm using at most $S$ ordinary samples from $q\in r(P)$ that estimates $f(q)$ to accuracy $\veps$ with confidence at least $1-2\delta$, contradicting the lower bound on $r(P)$ proved above.
	Hence, no algorithm using $S$ uniform subset-element-samples can estimate $f(p')$ to accuracy $\veps$ with confidence $1-\delta$.
\end{proof}

\begin{restatable}{cor}{stringSampLowerCor}\label{cor:stringSampLower}
    Let $\veps,\delta\leq\frac18$, $K=\Omega\Big(\frac{m^2+\log^2(1/\delta)}{\veps^4\delta}\Big)$, $n\ge\log_2(mK)$, and suppose $\{0,1\}^n$ is partitioned into $m$ bins $B_1,\dots,B_m$ of sizes differing by at most $1$.
    Suppose that $p'=(\frac{k_1}{K}, \frac{k_2}{K}, \ldots, \frac{k_m}{K})$ is an unknown distribution of rational probabilities on $m$ elements, and we are given a uniformly random subset $I\subseteq\{0,1\}^n$ with $|I|= K$, so that
    $p'_i \;=\; \frac{|I\cap B_i|}{|I|}$.
    Then any algorithm that outputs an estimate $q\in [0,1]^m$ of $p'$ within total variation distance $\veps$ with probability at least $1-\delta$, must use at least $\Omega\Big(\frac{m+\log(1/\delta)}{\veps^2}\Big)$ i.i.d.\ uniform random samples from elements of $I$.
\end{restatable}
\begin{proof}
    Apply \autoref{lem:stringSampLower} for the problem of learning a distribution on $m$ elements to total variation distance $\veps$, i.e., set $f$ to the identity function, $L=1$, and $S=\Omega\bigg(\frac{m+\log(1/\delta)}{\veps^2}\bigg)$ according to the distribution learning lower bounds \cite[Theorem 1.3]{cannone2022DistributionTestingSurvey}, \cite[Lecture 4]{diakonikolas2019AdvancedLearningTheoryLectureNotes}.
\end{proof}

\subsection{Hybrid method for bounding small probability advantages}\label{apx:smallHybLow}

In order to formalize the intuition of the limited information retrieval from the queries, we use the well-known quantum query lower bound technique called the \emph{hybrid method}. While the following general lower bound result and its proof is based on the standard hybrid method, unlike the usual formulations, it also considers the low probability regime. This is a standard technique, but we could not find a readily applicable low-probability variant in the literature, so for completeness, we present a proof here. This technique was originally introduced for proving a lower bound for quantum search by Bennett et al.~\cite{bennett1997QSearchLowerBound}, and can be viewed as a simplified version of the adversary method \cite{ambainis2000AdversaryMethod,hoyer2005AdversaryReview,lee2011QQueryCompStateConv,belovs2015GeneralAdv}.
Our proof closely follows the presentation of the hybrid method in Montanaro's lecture notes \cite[Section 1]{montanaro2011LectureNotes}, and that of~\cite{gilyen2017OptQOptAlgGrad}.

\begin{restatable}{lem}{hybMetArb}\label{lem:smallHybLow}
	Let $B$ be a (finite) set of labels and let $\CH:=\mathrm{Span}(\ket{b}\colon b\in B)$ be a Hilbert space. For a function $f:B\rightarrow \BR$ let $\mathrm{O}_{\!f}$ be the phase oracle acting on $\CH$ such that $$\mathrm{O}_{\!f}\colon\ket{b}\to \exp(\pi \ri f(b))\ket{b} \quad \text{ for every } b\in B.$$
	Suppose that $\CF$ is a finite set of functions $B\rightarrow \{0,1\}$. Let $\vrho_{\CF}$ denote the output of a quantum algorithm that makes $T$ queries to a (controlled) phase oracle $\mathrm{O}_{\!f}$ for a fixed but random $f$ chosen according to some distribution on $\CF$, and let $\vrho_{f_0}$ denote its output when all queries are replaced by $\mathrm{O}_{\!f_0}$. 
	The trace distance of these states has
	\[ \frac{1}{2}\nrm{\vrho_{\CF}-\vrho_{f_0}}_1 \leq T \max_{b\in B}\sqrt{\underset{\CF}{\BE}[|\exp(\pi \ri f(b))-\exp(\pi \ri f_0(b))|^2]}
	\leq 2T \max_{b\in B} \sqrt{\Pr[f(b)\neq f_0(b)]}.\]
\end{restatable}
\begin{proof}According to the church of the larger Hilbert space, more precisely due to the triangle and data processing inequalities, we can assume without loss of generality that the algorithm's starting state is pure $\ket{\vec{0}}$ and the algorithm is unitary (i.e., does not involve measurements).
       
	Suppose that $\CF=\{f_j\colon j\in \{1,\ldots,K\}\}$. Let $\CA_j$ denote the algorithm which uses phase oracle $\mathrm{O}_{\!f_j}$ and let $\ket{\psi_j}:=\CA_j\ket{\vec{0}}$ denote the final state of the algorithm. 
	Since $\frac{1}{2}\nrm{\ketbra{\psi_j}{\psi_j}-\ketbra{\psi_0}{\psi_0}}_1\leq \nrm{\ket{\psi_j}-\ket{\psi_0}}$, due to the triangle inequality it suffices to bound 
	\begin{align}\label{eq:expectedl2}
		\frac{1}{2}\nrm{\vrho_{\CF}-\vrho_{f_0}}_1
		=\frac{1}{2}\nrm{\BE[\ketbra{\psi_j}{\psi_j}]-\ketbra{\psi_0}{\psi_0}}_1
		\leq \BE[\frac{1}{2}\nrm{\ketbra{\psi_j}{\psi_j}-\ketbra{\psi_0}{\psi_0}}_1]
		\leq \BE[\nrm{\ket{\psi_j}-\ket{\psi_0}}].
	\end{align}
	
	We can assume without loss of generality that $\mathrm{O}_f$ is a controlled oracle, i.e., there is a known element $b_0\in B$ such that $f_j(b_0)=0$ for all $j\in \{0,1,\ldots, K\}$.
	In general the quantum algorithm might use some workspace $\CW=\underset{w\in W}{\mathrm{Span}}\{\ket{w}: w\in W\}$ along with the Hilbert space $\CH$. With a slight abuse of notation, from now on we account for the presence of this additional Hilbert space by replacing $\mathrm{O}_f$ by $\mathrm{O}_f\otimes I_\CW$ and $B$ by $B\times W$, so that the elements of $\CH\otimes \CW$ can be labeled by the elements of $B$.
	
	Let us define for $t\in \{0,1,\ldots, T\}$
	$$\ket{\psi_j^{(t)}}:=\left(\prod_{\tau=1}^{t}\vU_\tau\mathrm{O}_{\!f_j}\right)\vU_0\ket{\vec{0}},$$
	the state of algorithm $\CA_j$ after making $t$ queries. We now prove by induction that for all $t\in \{0,1,\ldots, T\}$
	\begin{align}
		\label{eq:inductioninequality}
		\nrm{\ket{\psi_j^{(t)}}-\ket{\psi_0^{(t)}}}\leq \sum_{\tau=0}^{t-1} \nrm{(\mathrm{O}_{\!f_j}-\mathrm{O}_{\!f_0})\ket{\psi_0^{(\tau)}}}.
	\end{align}
	For $t=0$ the left-hand side is $0$, so the base case holds. Let us assume that \eqref{eq:inductioninequality} holds for $t-1$, we prove the inductive step as follows:
	\begin{align*}
		\nrm{\ket{\psi_j^{(t)}}-\ket{\psi_0^{(t)}}}
		&=\nrm{\vU_t\mathrm{O}_{\!f_j} \ket{\psi_j^{(t-1)}}-\vU_t\mathrm{O}_{\!f_0}\ket{\psi_0^{(t-1)}}}	\\
		&=\nrm{\mathrm{O}_{\!f_j} \ket{\psi_j^{(t-1)}}-\mathrm{O}_{\!f_0}\ket{\psi_0^{(t-1)}}}	 \tag{since norms are unitarily invariant}  	\\		
		&=\nrm{\mathrm{O}_{\!f_j} \left(\ket{\psi_j^{(t-1)}}-\ket{\psi_0^{(t-1)}}+\ket{\psi_0^{(t-1)}}\right)-\mathrm{O}_{\!f_0}\ket{\psi_0^{(t-1)}}}	   	\\	
		&\leq\nrm{\mathrm{O}_{\!f_j} \left(\ket{\psi_j^{(t-1)}}-\ket{\psi_0^{(t-1)}}\right)} + \nrm{(\mathrm{O}_{\!f_j} -\mathrm{O}_{\!f_0})\ket{\psi_0^{(t-1)}}}	   	\tag{triangle inequality}\\	
		&=\nrm{\ket{\psi_j^{(t-1)}}-\ket{\psi_0^{(t-1)}}} + \nrm{(\mathrm{O}_{\!f_j} -\mathrm{O}_{\!f_0})\ket{\psi_0^{(t-1)}}}	   	\\	
		&\leq \sum_{\tau=0}^{t-1} \nrm{(\mathrm{O}_{\!f_j}-\mathrm{O}_{\!f_0})\ket{\psi_0^{(\tau)}}}. \tag{by the induction hypothesis}
	\end{align*}
	
	Since $\ket{\psi_j}=\ket{\psi_j^{(T)}}$, we additionally have that
	\begin{equation*}
		\nrm{\ket{\psi_j}-\ket{\psi_0}}^2 
		\leq \left(\sum_{\tau=0}^{T-1} \nrm{(\mathrm{O}_{\!f_j}-\mathrm{O}_{\!f_0})\ket{\psi_0^{(\tau)}}}\right)^{\!\!2}
		\leq T\sum_{\tau=0}^{T-1} \nrm{(\mathrm{O}_{\!f_j}-\mathrm{O}_{\!f_0})\ket{\psi_0^{(\tau)}}}^2,
	\end{equation*}
	where the last inequality uses the Cauchy-Schwarz inequality.
	By Jensen's inequality, we also have 
	\begin{align}
		\BE[\nrm{\ket{\psi_j}-\ket{\psi_0}}] 
		= \BE[\sqrt{\nrm{\ket{\psi_j}-\ket{\psi_0}}^2}]
		&\leq \sqrt{T\sum_{\tau=0}^{T-1}\BE\left[ \nrm{(\mathrm{O}_{\!f_j}-\mathrm{O}_{\!f_0})\ket{\psi_0^{(\tau)}}}^2\right]\nonumber}\\
		&\leq T\sqrt{\max_{\tau\in\{0,\ldots,T-1\}}\BE\left[ \nrm{(\mathrm{O}_{\!f_j}-\mathrm{O}_{\!f_0})\ket{\psi_0^{(\tau)}}}^2\right]}.\label{eq:Jensen}
	\end{align}
	We now upper bound the right-hand side of Eq.~\eqref{eq:Jensen} for an arbitrary pure state $\ket{\psi}$ to conclude the proof.
	\begin{align*}
		\BE\left[\nrm{(\mathrm{O}_{\!f_j}-\mathrm{O}_{\!f_0})\ket{\psi}}^2\right]
		&=\BE\left[ \nrm{\left(\sum_{b\in B}\ketbra{b}{b}\right)(\mathrm{O}_{\!f_j}-\mathrm{O}_{\!f_0})\left(\sum_{b'\in B}\ketbra{b'}{b'}\right)\ket{\psi}}^2\right]\\
		&=\BE\left[\nrm{\sum_{b\in B}\ketbra{b}{b}(\mathrm{O}_{\!f_j}-\mathrm{O}_{\!f_0})\ketbra{b}{b}\ket{\psi}}^2 \tag{since $\braket{b |\mathrm{O}_{\!f_j}|b'}=0$ for $b\neq b'$}\right]\\			
		&=\sum_{b\in B} \left|\braket{b|\psi}\right|^2 \BE\left[ \left|\bra{b}(\mathrm{O}_{\!f_j}-\mathrm{O}_{\!f_0})\ket{b}\right|^2\right]\\					
		&\leq \max_{b\in B} \BE\left[ \left|\bra{b}(\mathrm{O}_{\!f_j}-\mathrm{O}_{\!f_0})\ket{b}\right|^2\right]\\					
		&= \max_{b\in B} \BE\left[|\exp(\pi \ri f(b))-\exp(\pi \ri f_0(b))|^2\right]\\ 
		&\leq \max_{b\in B} 4 \Pr[f(b)\neq f_0(b)].
	\end{align*}
	Combining this upper bound with Eqs.~\eqref{eq:expectedl2} and \eqref{eq:Jensen}, we obtain the desired inequality.
\end{proof}

\subsection{Assembling the proof of \autoref{thm:tightSample}}\label{apx:proofTight}
\sampleLowerBound*
\begin{proof}[Proof of \autoref{thm:tightSample}]
	First, let us establish the lower bound in the absence of queries to $H$, assuming that $n=\lfloor\log_2\lfloor\log_2(M+1)+1\rfloor \rfloor$. We will use the folklore bound that learning distributions on $m$ elements to total variation distance $\veps$ requires $\Omega\Big(\frac{m+\log(1/\delta)}{\veps^2}\Big)$ samples \cite[Theorem 1.3]{cannone2022DistributionTestingSurvey} (even if we are promised that all $p_i\in[\frac{1}{2m},\frac{3}{2m}]$~\cite[Lecture 4]{diakonikolas2019AdvancedLearningTheoryLectureNotes}). Set $m=2^n$, and let us use the observables $\{\chi_s\colon s\subseteq \{0,1\}^n\}$ that correspond to the indicator functions of the $2^n$-element configuration space, i.e., for all $n$-bit configuration $c$ we have $\chi_s(c)=1$ if and only if $c\in s$. Let $v_s:=\BE[\chi_s]=\Pr[c\in s]$ denote the expectation value of $\chi_s$. Since $v_{\bar{s}}=1-v_s$ for $\bar{s}=\{0,1\}^n\setminus s$, it follows that it suffices to choose half of the subsets $\emptyset\subsetneq s\subsetneq \{0,1\}^n$ for estimating $\BE[\chi_s]=\Pr[c\in s]$ within precision $\veps$ to provide an $\veps$-precise estimate of $v_s$ for all $s\subseteq\{0,1\}^n$. 

    Thus, by estimating $2^{m-1}-1\leq M$ (indicator) observables to precision $\veps$ we recover $v_s$ values such that the (Gibbs) distribution $\rho=(p_1,p_2,\ldots,p_m)$ satisfies $|v_s - \sum_{c \in s} p_c|\leq \veps$. By the triangle inequality, for all distribution $q$ satisfying $|v_s - \sum_{c \in s} q_c|\leq \veps$ we have $|\sum_{c\in s}p_c - \sum_{c\in s}q_c|\leq 2\veps$, i.e., we can recover a probability distribution $(q_1,q_2,\ldots,q_m)$ on $m=2^n$ elements that is within total variation distance $2\veps$ of the (Gibbs) distribution.
	
	Note that when $p_i\in[\frac{1}{2m},\frac{3}{2m}]$, at $\beta=1$, the corresponding diagonal Hamiltonian with $H_i:=-\ln(m p_i)$ satisfies $|H|\leq1$. The distribution learning lower bounds \cite[Theorem 1.3]{cannone2022DistributionTestingSurvey}, \cite[Lecture 4]{diakonikolas2019AdvancedLearningTheoryLectureNotes} then imply that $\Omega\Big(\frac{m+\log(1/\delta)}{4\veps^2}\Big)=\Omega\Big(\frac{\log(M/\delta)}{\veps^2}\Big)$ Gibbs samples are necessary for solving this estimation task, however one can also directly solve the problem by querying all $m=\bigO{\log(M)}$ diagonal entries of the classical Hamiltonian~$H$.
	
	In the general case in order to prove that a subexponential number of queries to $H$ does not help significantly reducing the sample complexity, we will consider a carefully crafted family of classical Hamiltonians $H_f\in\{0,1\}^{2^n\times 2^n}$ whose diagonal elements are associated to a Boolean function $f\colon \{0,1\}^n\mapsto \{0,1\}$. The corresponding Gibbs state $\rho_\beta$ assigns probability $\rho_\beta(b) = e^{-\beta f(b)}/Z_\beta$ to configuration $b$ where  $Z_\beta = \sum_{b\in\{0,1\}^n}e^{-\beta f(b)}$. 
	
	By \autoref{lem:tempTVBound} $\beta\geq \ln(2)n + \ln(\frac{8S}{\veps \delta})=\Theta(n)$ ensures that \[\|\rho_\beta - \rho_\infty\|_{\mathrm{TV}}\leq \min\Big\{\frac{\veps}4,\frac{\delta}{2S}\Big\},\] implying that $\mathbb{E}_{\rho_\beta}[O_i]-\mathbb{E}_{\rho_\infty}[O_i]\leq \frac{\veps}{2}$, therefore estimating the expectation values of one of the distributions $\rho_\beta, \rho_\infty$ to precision $\frac{\veps}{2}$ suffices for estimating those of the other to precision $\veps$. Similarly, as long as one uses at most $S$ samples and the protocol succeeds with probability at least $1- \frac{\delta}{2}$ using one of the distributions, then the same protocol also succeeds with probability at least $1- \delta$ when samples of the other Gibbs state are provided.
	
	So if we can estimate all expectation values with respect to $\rho_\beta$ with precision $\frac{\veps}{2}$ and success probability at least $1-\frac{\delta}{2}$ using $S$ Gibbs samples $\rho_\beta$ and $Q$ queries to $f$, then we can also solve with the same algorithm the same estimation problem with respect to $\rho_\infty$ but with a doubled $\veps$ and $\delta$ parameter. Therefore, it suffices to prove a lower bound for the latter problem.
	
	To connect this scenario to the previous lower bound we will work with a distribution $p$ on $m=\lfloor\log_2(M+1)+1\rfloor$ elements induced by the Gibbs state $\rho_\infty$, and show how to use the estimation of $M$ expectation values of the Gibbs state $\rho_\infty$ to learn the induced distribution $p$. 
	For this we partition the configuration space to bins of almost equal size $\lceil\frac{2^n}{m}\rfloor$. We denote the indicator of the $\ell$-th bin by $\chi_\ell$. For each $b\in\{0,1\}^m$ we define the observable
	\begin{equation}\label{eq:splittingObs}
		O_b=\sum_{\ell\in [m]}b_\ell\chi_\ell.
	\end{equation}
	The same argument that we used at the beginning of the proof shows that successfully providing an estimate $v_b$ to all observable expectation values $\mathbb{E}[O_b]=\sum_{\ell\in [m]}b_\ell\mathbb{E}[\chi_\ell]$ of precision $\veps$ implies that we can recover a probability distribution $q=(q_1,q_2,\ldots,q_m)$ on $m$ elements that is within total variation distance $2\veps$ of the distribution $p=(\mathbb{E}[\chi_1],\mathbb{E}[\chi_2],\ldots,\mathbb{E}[\chi_m])$ induced by our Gibbs state $\rho_\infty$. 
	
	Like before, we use a reduction to the above cited distribution learning lower bound \cite[Theorem 1.3]{cannone2022DistributionTestingSurvey}, \cite[Lecture 4]{diakonikolas2019AdvancedLearningTheoryLectureNotes}. Let us set $f^{-1}(0):=\left\lfloor\frac{\delta^2}{8m}2^{\frac n3}\right\rfloor$ for a sufficiently large $n\geq \Omega\left(\log\left(\frac{m}{\veps\delta}\right)\right)$. The motivation for this choice is the following: if $f$ is a uniformly random function satisfying $f^{-1}(0)=\Omega\left(\frac{\log^2(M/\delta)}{\veps^4\delta}\right)$, then as in \autoref{cor:stringSampLower}, it is not difficult to see that sampling $\ll\Big(\frac{\log(M/\delta)}{\veps^2}\Big)$ uniformly random elements of $f^{-1}(0)$, i.e., taking that many samples from $\rho_\infty$, reveals essentially no more information from the distribution $p_\ell=|\chi^{-1}_\ell(1)\cap f^{-1}(0)|/|f^{-1}(0)|$ than directly taking the same amount of i.i.d.\ samples from $p$. Thus, for such $f$, the same asymptotic lower bound $\Omega\Big(\frac{\log(M/\delta)}{\veps^2}\Big)$ holds for the number of required Gibbs samples.
	
	Now we complete the argument by bounding the advantage given by the (quantum) queries. Because the samples are independent of everything the estimation protocol does, we can assume without loss of generality that it first receives the i.i.d.\ samples $(i_i,i_2,\ldots,i_S)$ of $\rho_\infty$ and subsequently makes $Q$ queries~to~$f$. The main idea of the lower bound proof is to choose an arbitrary (rounded) distribution $p\in\Delta_m$ and a corresponding random\footnote{Note that we use here the same uniform distribution over the different realizations of a distribution for which the hardness is established by \autoref{cor:stringSampLower}, so the classical and quantum lower bound arguments match up.} $f$ with Hamming weight $ 2^n-\left\lfloor\frac{\delta^2}{8m}2^{\frac n3}\right\rfloor$ so that once the algorithm receives the samples $\rho_\infty$ (which are just uniformly random preimages of $0$) it cannot gain any more information about the set $f^{-1}(0)$ with probability greater than $\delta$, when it makes $Q\le 2^{\frac{n}{3}}$ queries. The intuition behind this lower bound is that using the (optimal) Grover search algorithm with $Q\leq 2^{\frac{n}{3}}$ queries cannot even provide a single yet unseen element $f^{-1}(0)$ with probability greater than $\delta$.
	
	We formalize our intuition about the limited information retrieval from the queries using a low-probability version of the \emph{hybrid method} proved in~\autoref{lem:smallHybLow}, which we recite here for better readability:

	\hybMetArb*
	
	Let $f_0$ be the function that is all-$1$ except $f(i_j)=0$ for the indices $i_j$ that were revealed by the $S$ samples. Let $\CF$ be the conditional distribution after seeing the samples, i.e., the uniform distribution over functions $f$ that evaluate to $0$ on the $S$ sampled indices, have Hamming weight as above, and correspond to the fixed (unknown) $p\in\Delta_m$. Then, for all $i\notin \{i_j\colon j\in [S]\}$ the bit $f(i)$ equals $0$ with probability at most $\frac{|f^{-1}(0)|\cdot p_\ell}{\lfloor2^n/m\rfloor}\leq \frac{\delta^2}{4} 2^{\frac{-2n}{3}}$, since the more zeros were revealed from bin $\ell$ by the initial samples, the less likely it becomes for other indices in that bin to evaluate to $0$.
	If the algorithm makes at most $2^{\frac{n}{3}}$ queries, then according to \autoref{lem:smallHybLow}, the algorithm's output cannot change by more than probability $\delta$ when we replace the true $f$ queries by the known dummy $f_0$. Thus, if the protocol only succeeds with probability $<1-2\delta$ after receiving the samples and running the algorithm using dummy $f_0$ queries, it would not be able to achieve success probability at least $1-\delta$ even using the proper $f$ (quantum) queries. However, since the dummy $f_0$ queries reveal absolutely no information about the distribution, to achieve such high success probability we need the stated number of samples. 
\end{proof}

\end{document}